\documentclass[conference,letterpaper]{IEEEtran}
\IEEEoverridecommandlockouts

\usepackage{tikz-cd}
\usepackage{graphicx}
\usepackage[svgnames]{xcolor}
\usepackage{balance}
\usepackage[utf8]{inputenc}
\usepackage[T1]{fontenc}
\usepackage{url}
\usepackage[cmex10]{amsmath} 
\usepackage[url,hyperrefblack,notheorems,IEEEtran]{research17} 
\usepackage{amsmath,amssymb,amsfonts,amsbsy}

\usepackage{soul} 
\usepackage[colorinlistoftodos,textsize=footnotesize]{todonotes} 
\colorlet{darkgreen}{DarkGreen}

\usepackage{amsthm} 
\newtheorem{theorem}{\mytheoremname}
\newtheorem{lemma}{\mylemmaname}
\newtheorem{corollary}{\mycorollaryname}

\newtheorem{proposition}{\mypropositionname}

\newtheorem{definition}{\mydefinitionname}
\newtheorem{remark}{\myremarkname}
\newtheorem{example}{\myexamplename}

\usepackage[capitalize]{cleveref}
\allowdisplaybreaks

\usepackage{booktabs} 
\usepackage{array}    
\usepackage{tabularx} 
\usepackage{multirow}
\usepackage{placeins}
\usepackage{enumitem}
\usepackage[lined,boxed,commentsnumbered,linesnumbered,algoruled]{algorithm2e}
\DontPrintSemicolon
\SetKw{Continue}{continue}

\usepackage{tikz}
\usepackage[caption=false,font=footnotesize]{subfig}
\newcommand{\Qchi}{\chi^{\prime}_{\tn{Q}}}
\newcommand{\Qgap}{\gamma_{\tn{Q}}}
\newcommand{\Aut}{\ensuremath{\operatorname{Aut}}}
\newcommand{\quotient}[2]{\ensuremath{#1/#2}} 

\title{Local Automorphism-Aware Syndrome Compilation for General Quantum LDPC Codes\thanks{\IEEEauthorrefmark{1}The first two authors share first authorship.}\thanks{This work was supported in part by the Research Council of Norway (RCN) under the NISQEC project (grant no.~$357698$) and the Centre for Quantum Communication Networks and Applications (QCNA)   (grant no.~$361582$).}
}
\author{\IEEEauthorblockN{Eugenio Durazo Rocha\IEEEauthorrefmark{1}\IEEEauthorrefmark{2}\IEEEauthorrefmark{3}, Olai {\AA}.~Mostad\IEEEauthorrefmark{1}\IEEEauthorrefmark{2}\IEEEauthorrefmark{3}, Hsuan-Yin Lin\IEEEauthorrefmark{2}, and Eirik Rosnes\IEEEauthorrefmark{2}}
\IEEEauthorblockA{\IEEEauthorrefmark{2}Simula UiB, N--5006 Bergen, Norway\\
\IEEEauthorrefmark{3}Department of Informatics, University of Bergen, N--5020 Bergen, Norway\\
Emails: \{eugenio, olai, lin, eirikrosnes\}@simula.no}}

\begin{document}
\bstctlcite{IEEEtran:default} 
\maketitle

\begin{abstract}
  Low-depth syndrome extraction for Calderbank--Shor--Steane (CSS) quantum low-density parity-check codes can be formulated as a \emph{proper ordered edge-coloring} problem subject to quantum parity constraints. A proper edge-coloring of the CSS Tanner graph ensures that each data or ancilla qubit participates in at most one two-qubit gate per layer, but does not guarantee a valid interleaving of the $\textnormal{X}$- and $\textnormal{Z}$-check measurements as for every overlapping $\textnormal{X}/\textnormal{Z}$ check pair, the number of shared data qubits on which the $\textnormal{X}$ interaction precedes the $\textnormal{Z}$ interaction must be even. The minimum number of colors in a proper ordered edge-coloring satisfying the quantum parity constraints equals the minimum two-qubit depth when each stabilizer check is measured with a single ancilla.
  
  We introduce \emph{local automorphism-aware syndrome compilation (LocalASC)}, which reduces the constraint system to edge-orbit variables under a subgroup of the Tanner graph automorphisms and lifts each feasible orbit assignment to the full graph. Although the $6$-layer degree lower bound is unattainable for the published weight-$6$ IBM bivariate bicycle codes, we show that this is not universal among two-block CSS codes. Among code instances for which the maximum check weight equals the maximum Tanner graph degree, LocalASC finds depth-optimal syndrome-extraction schedules for several two-block CSS codes with odd component weights, including instances with unequal odd weights. We also obtain lower-bound-saturating syndrome-extraction schedules for several quantum Tanner codes satisfying the same degree condition. To obtain the subgroups used by LocalASC without computing the full automorphism group of the Tanner graph, we construct translation subgroups for two-block group-algebra CSS codes over abelian groups. For quantum Tanner codes, we give conditions under which square-complex symmetries extend to Tanner graph automorphisms.
\end{abstract}

\section{Introduction}
\label{sec:introduction}

Low-depth syndrome extraction is a central implementation problem for quantum low-density parity-check (qLDPC) codes. Strikis \emph{et al.} developed a high-performance syndrome-extraction framework based on left-right circuits, extended this noninterleaving construction to arbitrary Calderbank--Shor--Steane (CSS) codes~\cite{CalderbankShor96_1, Steane96_1}, and optimized both qubit idling time and circuit-level effective distance~\cite{StrikisBrowneBeverland26_1sub}. Their framework maintains low depth by staggering the measurements of $\tn{X}$- and $\tn{Z}$-checks without interleaving their gates. For a stabilizer code~\cite{Gottesman97_1, CalderbankRainsShorSloane98_1} with a fixed set of measured stabilizer generators, let $\graph{G}$ denote the corresponding Tanner graph, let $\Delta(\graph{G})$ denote its maximum degree, and let $\lambda^\star_{\tn{SEC}}$ denote the minimum two-qubit depth over all valid syndrome-extraction circuits for these generators. Since the two-qubit gates incident on any Tanner graph vertex cannot be executed simultaneously, they must occupy distinct circuit layers. Hence, $\lambda^\star_{\tn{SEC}}\geq\Delta(\graph{G})$. This bound is not always attainable because $\tn{C}\mat{X}$ and $\tn{C}\mat{Z}$ gates for overlapping $\tn{X}$- and $\tn{Z}$-checks cannot be ordered arbitrarily. Zhang \emph{et al.} incorporated these ordering constraints into the \emph{Auto-Stabilizer-Check (ASC)} compiler~\cite{Zhang-etal26_1sub}. For the bivariate bicycle (BB) codes of Bravyi \emph{et al.}, ASC proves that depth $6$ is unattainable even though $\Delta(\graph{G})=6$~\cite{Bravyi-etal24_1, Zhang-etal26_1sub}.

We recast the syndrome extraction problem as \emph{quantum-constrained edge coloring (QCEC)}, from a graph-theoretic perspective. Each Tanner graph edge is assigned a color that specifies its two-qubit circuit layer. Proper edge coloring prevents gate collisions, while the quantum parity constraints enforce a valid ordering for every overlapping $\tn{X}/\tn{Z}$ check pair. Given a CSS code with a fixed pair of parity-check matrices (PCMs) $(\mat{H}_{\tn{X}},\mat{H}_{\tn{Z}})$ for stabilizer measurements, the minimum number of colors defines the \emph{quantum-constrained edge chromatic number} $\Qchi(\mat{H}_{\tn{X}},\mat{H}_{\tn{Z}})$. This number depends on the chosen PCMs, not only on the stabilizer group they generate (see Theorem~\ref{thm:qcec-depth}).


Furthermore, we introduce \emph{local automorphism-aware syndrome compilation (LocalASC)}. Let $\Aut(\graph{G})$ denote the group of CSS Tanner graph automorphisms preserving the $\tn{X}$-check, $\tn{Z}$-check, and data-qubit vertex classes. For $\set{H}\leq\Aut(\graph{G})$, LocalASC assigns one tick variable to each edge orbit while retaining all QCEC constraints on the full Tanner graph. By Theorem~\ref{thm:orbit-invariant-schedule-lifting}, each feasible orbit assignment lifts to an $\set{H}$-invariant ASC-feasible schedule, and every such schedule determines a unique orbit assignment. The method applies to any chosen subgroup $\set{H}\leq\Aut(\graph{G})$, whether found using a generic graph-automorphism routine or supplied by the code construction. Algorithm~\ref{alg:LocalASC} summarizes the complete procedure and independently verifies that the lifted schedule satisfies all QCEC constraints on the full Tanner graph.

\begin{table}[t!]
  \caption{LocalASC versus ASC at depth $w=\Delta(\graph{G})$. Both methods were run with a $30$ minutes time limit for each instance. Runtimes are averaged over $10$ runs, and instances exceeding the time limit were not repeated. Bold code parameters mark instances for which ASC found no schedule within the time limit.}
  \vspace{-3ex}
  \label{tab:LocalASC-vs-ASC_depth6}
  \begin{center}
    \begingroup
    \setlength{\tabcolsep}{3pt}
    \begin{tabular}{@{}lccccl@{}}
      \toprule
      $[[n,k,d]]$
      & $\lambda^\star_{\tn{SEC}}$
      & LocalASC
      & $\ecard{\set{H}^{\star}}$
      & ASC
      & References
      \\
      \midrule
      \multicolumn{6}{l}{\itshape Two-block codes, $(w_{\tn{A}},w_{\tn{B}})=(3,3)$}
      \\
      $[[12,4,2]]$                 & 6 & 45\,ms  & 3  & $<0.1$\,s  & \cite{LiangLiuSongChen25_1} \\
      $[[24,4,4]]$                 & 6 & 25\,ms  & 3  & 0.3\,s     & \cite{LiangLiuSongChen25_1} \\
      $[[28,6,4]]$                 & 6 & 12\,ms  & 7  & 7\,s       & \cite{LiangLiuSongChen25_1} \\
      $[[36,4,4]]$                 & 6 & 38\,ms  & 3  & 15\,s      & \cite{PostemaKokkelmans25_1sub} \\
      $[[36,4,6]]$                 & 6 & 66\,ms  & 3  & 12\,s      & \cite{Wang-etal26_1} \\
      $[[48,4,8]]$                 & 6 & 54\,ms  & 3  & 33\,s      & \cite{LiangLiuSongChen25_1} \\
      $[[48,4,8]]$                 & 6 & 73\,ms  & 3  & 37\,s      & \cite{AydinTamoBarg26_1sub} \\
      $\bm{[[108,12,6]]}$      & 6 & 769\,ms & 27 & $>30$\,min & \cite{WangMueller24_1sub} \\
      \midrule
      \multicolumn{6}{l}{\itshape Two-block codes, $(w_{\tn{A}},w_{\tn{B}})=(3,5)$}
      \\
      $\bm{[[84,14,10]]}$      & 8 & 43\,ms  & 21 & $>30$\,min & \cite{LinPryadko24_1} \\
      $\bm{[[96,12,10]]}$      & 8 & 10\,s   & 12 & $>30$\,min & \cite{LinPryadko24_1} \\
      $\bm{[[112,12,12]]}$     & 8 & 11\,s   & 14 & $>30$\,min & \cite{LinPryadko24_1} \\
      $\bm{[[120,8,8]]}$       & 8 & 400\,ms & 15 & $>30$\,min & \cite{LinPryadko24_1} \\
      \midrule
      \multicolumn{6}{l}{\itshape Quantum Tanner codes}
      \\
      $[[45,7,4]]$                & 6 & 49\,ms  & 5  & 150\,s     & \cite{MostadRosnesLin26_1sub, nisqec26_1} \\
    $\bm{[[180,26,6]]}$     & 6 & 549\,ms & 5  & $>30$\,min & \cite{MostadRosnesLin25_1, nisqec26_1} \\
    $[[325,14,9]]$              & 9 & 185\,ms & 13 & 548\,s     & \cite{MostadRosnesLin26_1sub, nisqec26_1} \\
      $\bm{[[666,4,d\geq13]]}^a$ & 9 & 3.1\,s  & 37 & $>30$\,min & \cite{MostadRosnesLin26_1sub, nisqec26_1} \\
      $\bm{[[666,8,d\geq13]]}^a$ & 9 & 1.4\,s  & 37 & $>30$\,min & \cite{MostadRosnesLin26_1sub, nisqec26_1} \\
      $[[1225,225,6]]$            & 12& 640\,ms & 25 & 566\,s     & \cite{MostadRosnesLin26_1sub, nisqec26_1} \\
      \bottomrule
    \end{tabular}
    \endgroup
  \end{center}
  \footnotesize{$^a$The minimum distance lower bound was calculated using the classical algorithm in~\cite{RosnesYtrehus09_1, RosnesYtrehusAmbrozeTomlinson12_1} as outlined in \cite[Sec.~V]{MostadRosnesLin25_1}.}
  \vspace{-4ex}
\end{table}

For a fixed pair of measured CSS PCMs, let $w$ denote the maximum check weight. Table~\ref{tab:LocalASC-vs-ASC_depth6} compares LocalASC with ASC on selected two-block CSS codes~\cite{LinPryadko24_1, Bravyi-etal24_1, AydinTamoBarg26_1sub} and quantum Tanner codes~\cite{LeverrierZemor22_1, LeverrierZemor23_1, LeverrierRozendaalZemor25_1sub} satisfying $w=\Delta(\graph{G})$. The LocalASC and ASC columns report average runtimes, measured as described in Section~\ref{sec:numerical-results}, while $\ecard{\set{H}^{\star}}$ denotes the order of the subgroup used to obtain the corresponding schedule. For every completed instance in the table, LocalASC effectively finds a syndrome-extraction schedule of depth $w$, thereby attaining the Tanner graph degree lower bound and establishing
$\lambda^\star_{\tn{SEC}}=w=\Delta(\graph{G})$.

For both the two-block and quantum Tanner code instances, we obtained the automorphism subgroups using the same generic graph-automorphism routine, without using their construction-specific algebraic structure. The average LocalASC runtime is at most $11$ seconds for every instance in Table~\ref{tab:LocalASC-vs-ASC_depth6}. For eight instances, including the $[[108,12,6]]$ two-block code and the $[[180,26,6]]$ quantum Tanner code, ASC finds no schedule within the $30$ minutes time limit. 
LocalASC typically finds a schedule within seconds, whereas a single ASC run may exceed $30$ minutes. This makes it practical to search many PCMs for those that admit a schedule of depth $\Delta(\graph{G})$. We also applied LocalASC to other PCMs, including quantum Tanner codes from~\cite{LeverrierRozendaalZemor25_1sub, WangLiuLiKubicaGu26_1sub}. These instances are not listed in Table~\ref{tab:LocalASC-vs-ASC_depth6}. For many of them, including several with $w\neq\Delta(\graph{G})$, LocalASC finds a schedule of depth $\Delta(\graph{G})$, which is therefore depth-optimal. For instance, the $[[1225,32,d\geq 13]]$ quantum Tanner code constructed from the cyclic group $\set{C}_{50}$ (of order $50$)~\cite{MostadRosnesLin26_1sub, nisqec26_1} has $\Delta(\graph{G})=14>w=12$.\footnote{The minimum distance lower bound was calculated using the classical algorithm in~\cite{RosnesYtrehus09_1, RosnesYtrehusAmbrozeTomlinson12_1} as outlined in~\cite[Sec.~V]{MostadRosnesLin25_1}.} LocalASC finds a depth-$14$ schedule in $309$ seconds on average, whereas ASC finds no schedule within the $30$ minutes time limit. Our code is available at~\cite{LocalASC26_1}.

Our main application concerns general qLDPC codes. LocalASC identifies multiple code instances for which the lower bound on the Tanner graph degree is attainable. For two-block CSS codes with unequal odd component weights $w_{\tn{A}}$ and $w_{\tn{B}}$, the general coset-based construction of Aydin, Tamo, and Barg~\cite{AydinTamoBarg26_1sub} and the sandwich construction of Nguyen \emph{et al.}~\cite{NguyenRimbach-RussBosco26_1sub} both yield depth-$(w_{\tn{A}}+w_{\tn{B}}+1)$ schedules. For the instances identified here, LocalASC instead finds depth-$(w_{\tn{A}}+w_{\tn{B}})$ schedules, attaining the degree lower bound and therefore the optimal SEC depth. For quantum Tanner codes, LocalASC identifies additional medium-size code instances with previously unreported depth-optimal syndrome-extraction schedules beyond those considered by Nguyen \emph{et al.}~\cite{NguyenRimbach-RussBosco26_1sub}.

Independent of this work, Nguyen \emph{et al.} proposed a construction-based approach that partitions the Tanner graph edges according to the algebraic description of the code, solves the resulting reduced scheduling problem, and lifts the assignment to the full syndrome-extraction circuit~\cite{NguyenRimbach-RussBosco26_1sub}. They derive analytical schedules with depths close to the degree lower bound for lifted-product and balanced-product codes and report depth-optimal schedules for the quantum Tanner code instances considered in~\cite{LeverrierRozendaalZemor25_1sub}. Both their approach and LocalASC reduce the number of tick variables and lift a reduced assignment to the full Tanner graph, but the two reductions are different. Their edge partitions are derived from the specific constructions of these three code families and need not be edge orbits under a subgroup of $\Aut(\graph{G})$. Their analytical schedules also impose sufficient conditions, including termwise parity conditions and equalities between conjugate edge classes, that can be stronger than the original ASC ordering constraints~\cite[Eqs.~(11) and (21)]{NguyenRimbach-RussBosco26_1sub}. For any CSS code with a fixed pair of PCMs for stabilizer measurements, QCEC instead characterizes the minimum two-qubit depth without assuming a particular code family. LocalASC restricts QCEC to schedules invariant under one of the selected subgroups $\set{H}_1,\ldots,\set{H}_{\Gamma}\leq\Aut(\graph{G})$ by identifying the edge ticks within each orbit and retaining every original ASC constraint. These subgroups may be computed using generic graph-automorphism routines or obtained from the code construction.

For instance, consider the fixed pair of measured CSS PCMs for the $[[12,4,2]]$ code listed in Table~\ref{tab:LocalASC-vs-ASC_depth6}. The analytical construction of Nguyen \emph{et al.} gives a depth-$7$ schedule that satisfies their termwise parity and conjugate-edge conditions. For the same measured CSS PCMs, LocalASC uses an order-$3$ automorphism subgroup, yielding $24$ edge orbits with one tick variable assigned to each orbit. The resulting orbit model admits a valid depth-$6$ schedule. We were unaware of~\cite{NguyenRimbach-RussBosco26_1sub} while developing the present work.

Finally, Section~\ref{sec:constructing-automorphisms} presents construction-based automorphism subgroups for LocalASC. Algorithm~\ref{alg:constructH_abelian-2BGA} constructs and verifies translation subgroups for two-block group-algebra CSS codes over abelian groups. For quantum Tanner codes, Theorems~\ref{thm:aut-QT} and~\ref{thm:autQT-ab-Cay} describe Tanner graph automorphisms arising from square-complex symmetries. Theorem~\ref{thm:achievable-depth_qTanner-local-codes} gives sufficient conditions for quantum-valid ordered colorings of the local CSS codes to lift to a schedule satisfying the QCEC constraints on the full Tanner graph.

\section{Notation and Preliminaries}
\label{sec:preliminaries}

\subsection{Notation}
\label{sec:notation}

In this paper, vectors are denoted by bold lowercase letters, matrices by sans serif uppercase letters, sets and groups by calligraphic uppercase letters, and graphs and codes by script uppercase letters; for example, $\vect{a}$, $\mat{A}$, $\set{A}$, and $\graph{G}$, respectively. The identity element of a group is denoted by $1$, the identity map by $\operatorname{id}$, and $e$ is reserved for an edge of a graph. We write $\set{C}_{\ell}$ for a cyclic group of order $\ell$ and use additive notation for its elements when convenient. For a subset $\set{T}$ of a group, $\egen{\set{T}}$ denotes the subgroup generated by $\set{T}$. In particular, $\egen{g}$ denotes the cyclic subgroup generated by an element $g$. For a graph $\graph{G}$, $\set{V}(\graph{G})$ and $\set{E}(\graph{G})$ denote the vertex and edge sets, respectively. For $v\in\set{V}(\graph{G})$, $\set{N}(v)$ denotes the set of vertices adjacent to $v$, and $\delta(v)$ denotes the set of edges incident on $v$. For integers $a,b\in\Naturals\cup\{0\}$ with $a\le b$, let $[a:b]\eqdef\{a,a+1,\ldots,b\}$. For a statement, its indicator function is defined as $\eI{\tn{statement}}\eqdef 1$ if the statement is true and $0$ otherwise.

\subsection{Quantum CSS Codes}
\label{sec:quantum-css-codes}

Let $\mat{H}_{\tn{X}}\in\Field_2^{m_{\tn{X}}\times n}$ and $\mat{H}_{\tn{Z}}\in\Field_2^{m_{\tn{Z}}\times n}$ be the CSS PCMs satisfying
\begin{equation}
  \mat{H}_{\tn{X}}\trans{\mat{H}_{\tn{Z}}}=0.
  \label{eq:cssorth}
\end{equation}
The rows of $\mat{H}_{\tn{X}}$ and $\mat{H}_{\tn{Z}}$ specify the $\tn{X}$- and $\tn{Z}$-type checks that are measured. For each CSS code, we fix a pair of PCMs for stabilizer measurements because the same code may admit different generator sets, including redundant low-weight checks.

\begin{definition}[CSS Tanner graph]
  \label{def:CSS-Tanner-graph}
  The CSS Tanner graph $\graph{G}(\mat{H}_{\tn{X}},\mat{H}_{\tn{Z}})$ has vertex partition $\set{V}(\graph{G})=\set{V}_{\tn{X}}\sqcup\set{V}_{\tn{D}}\sqcup\set{V}_{\tn{Z}}$, where $\set{V}_{\tn{D}}=\{q_1,\ldots,q_n\}$ is the set of data vertices, and $\set{V}_{\tn{X}}$ and $\set{V}_{\tn{Z}}$ are the sets of $\tn{X}$- and $\tn{Z}$-check vertices, respectively. An edge $(x_i,q_j)$ is present if and only if ${(\mat{H}_{\tn{X}})}_{i,j}=1$, while an edge $(z_i,q_j)$ is present if and only if ${(\mat{H}_{\tn{Z}})}_{i,j}=1$.
  We denote the maximum degree of $\graph{G}$ by $\Delta(\graph{G})$.
\end{definition}

\subsection{Coset-Based Two-Block CSS Codes}
\label{sec:two-block}

A coset-based two-block code is specified by a finite group $\set{G}$, a subgroup $\set{K}\leq\set{G}$, and two group-algebra elements. Following~\cite[Secs.~II and III]{AydinTamoBarg26_1sub}, let
\begin{equation*}
  \set{N}_{\set{G}}(\set{K})\eqdef\{g\in\set{G}:g\set{K}=\set{K}g\}
\end{equation*}
be the normalizer of $\set{K}$ in $\set{G}$, and let $m=[\set{G}:\set{K}]$. The group $\set{G}$ acts from the left on the coset space $\quotient{\set{G}}{\set{K}}$, while $\set{N}_{\set{G}}(\set{K})$ acts from the right, according to
\begin{equation*}
  \mat{L}(g)(x\set{K})\eqdef(gx)\set{K},
  \qquad
  \mat{R}(h)(x\set{K})\eqdef(xh)\set{K}.
\end{equation*}
The right action is well defined because $h\in\set{N}_{\set{G}}(\set{K})$. Fix an ordering $\quotient{\set{G}}{\set{K}}=\{x_1\set{K},\ldots,x_m\set{K}\}$. The corresponding $m\times m$ permutation matrices are
\begin{equation*}
  [\mat{L}(g)]_{i,j}=\eI{\mat{L}(g)(x_j\set{K})=x_i\set{K}},
\end{equation*}
and
\begin{equation*}
  [\mat{R}(h)]_{i,j}=\eI{\mat{R}(h)(x_j\set{K})=x_i\set{K}}.
\end{equation*}
Let
\begin{equation*}
  a = \sum_{g\in\set{G}}a_g g\in\Field_2[\set{G}],
  \quad
  b = \sum_{h\in\set{N}_{\set{G}}(\set{K})}b_h h\in\Field_2[\set{N}_{\set{G}}(\set{K})].
\end{equation*}
Their supports are defined as $\supp(a)\eqdef\{g\in\set{G}:a_g\neq0\}$ and $\supp(b)\eqdef\{h\in\set{N}_{\set{G}}(\set{K}):b_h\neq0\}$, respectively. Extending the actions linearly, set
\begin{equation*}
  \mat{A}=\mat{L}(a)=\sum_{g\in\set{G}}a_g\mat{L}(g),
  \quad
  \mat{B}=\mat{R}(b)=\sum_{h\in\set{N}_{\set{G}}(\set{K})}b_h\mat{R}(h).
\end{equation*}
The left and right coset actions commute. Therefore, the binary coset-based two-block code $\set{Q}_{\set{G}}^{\set{K}}(a,b)$ with check matrices
\begin{equation}
  \mat{H}_{\tn{X}}=[\mat{A}\mid\mat{B}],
  \quad
  \mat{H}_{\tn{Z}}=[\trans{\mat{B}}\mid\trans{\mat{A}}]
  \label{eq:Hz-Hx_two-block-codes}
\end{equation}
has length $n=2m$ and satisfies $\mat{H}_{\tn{X}}\trans{\mat{H}_{\tn{Z}}}=0$ over $\Field_2$. When $\set{K}$ is normal, the construction reduces to a two-block group-algebra (2BGA) code over $\set{G}/\set{K}$. In particular, $\set{K}=\{1\}$ gives the usual 2BGA construction over $\set{G}$~\cite{LinPryadko24_1}. If the permutation matrices selected by $a$ and $b$ are distinct within their respective blocks, then the check weight is $w=\ecard{\supp(a)}+\ecard{\supp(b)}$.

\begin{example}[A {$[[6,2,2]]$} Code]
  \label{ex:622}
  Consider the binary coset-based two-block code with trivial subgroup $\set{K}=\{1\}$, cyclic group $\set{G}=\set{C}_3\eqdef\egen{r\mid r^3=1}$, and $a=b=1+r\in\Field_2[\set{C}_3]$. Here, $\set{N}_{\set{G}}(\set{K})=\set{G}$ and $\set{G}/\set{K}\cong\set{G}$, so this instance is the usual 2BGA code. Under the standard identification $\Field_2[\set{C}_3]\cong\Field_2[x]/(x^3-1)$, it can also be seen as a generalized bicycle (GB) code, denoted by $\tn{GB}[1+x,1+x]$~\cite[Sec.~IV-E]{LinPryadko24_1}. Order the cosets by the representatives $1,r,r^2$. The action of $r$ is represented by a cyclic permutation matrix, as defined in~\cite[Eq.~(43)]{LinPryadko24_1},
  \begin{equation*}
    \mat{P}=
    \begin{pmatrix}
      0&0&1\\
      1&0&0\\
      0&1&0
    \end{pmatrix}.
  \end{equation*}
  Consequently,
  \begin{equation*}
    \mat{A}=\mat{L}(a)=\mat{B}=\mat{R}(b)=\mat{I}+\mat{P}=
    \begin{pmatrix}
      1&0&1\\
      1&1&0\\
      0&1&1
    \end{pmatrix},
  \end{equation*}
  where $\mat{L}(a)=\mat{R}(a)$ because $\set{C}_3$ is abelian. This gives $\mat{H}_{\tn{X}}$ and $\mat{H}_{\tn{Z}}$ based on~\eqref{eq:Hz-Hx_two-block-codes}.
  
  The polynomial $h(x)=\gcd(1+x,1+x,x^3-1)=1+x$ has degree $1$, so the cyclic 2BGA dimension formula~\cite[Eq.~(44)]{LinPryadko24_1} gives $k=2\deg h=2$. Equivalently, both PCMs have rank $2$. A minimum-weight logical-operator search gives $d=2$. Hence, the code has parameters $[[6,2,2]]$. Its CSS Tanner graph has $\tn{X}$-check vertices $x_0,x_1,x_2$, $\tn{Z}$-check vertices $z_0,z_1,z_2$, and data vertices $A_0,A_1,A_2,B_0,B_1,B_2$.\footnote{Here, $A_i=q_{i+1}$ and $B_i=q_{i+4}$ for $i=0,1,2$, corresponding to the first and second column blocks of the CSS PCMs, respectively.} This graph has $24$ edges, is $4$-regular, and is illustrated in Fig.~\ref{fig:tanner622-uncolored}.
\end{example}

\begin{figure}[t!]
  \centering
  \resizebox{0.95\columnwidth}{!}{%
    \begin{tikzpicture}[every node/.style={font=\scriptsize}]
      \coordinate (X0) at (0,0); \coordinate (X1) at (0,-0.9); \coordinate (X2) at (0,-1.8);
      \coordinate (Z0) at (8,0); \coordinate (Z1) at (8,-0.9); \coordinate (Z2) at (8,-1.8);
      \coordinate (Q0) at (4,0.225); \coordinate (Q1) at (4,-0.225); \coordinate (Q2) at (4,-0.675);
      \coordinate (Q3) at (4,-1.125); \coordinate (Q4) at (4,-1.575); \coordinate (Q5) at (4,-2.025);
      \begin{scope}[black,line width=0.7pt]
        \draw (X0)--(Q0); \draw (X0)--(Q2); \draw (X0)--(Q3); \draw (X0)--(Q5);
        \draw (X1)--(Q0); \draw (X1)--(Q1); \draw (X1)--(Q3); \draw (X1)--(Q4);
        \draw (X2)--(Q1); \draw (X2)--(Q2); \draw (X2)--(Q4); \draw (X2)--(Q5);
        \draw[densely dashed] (Z0)--(Q0); \draw[densely dashed] (Z0)--(Q1); \draw[densely dashed] (Z0)--(Q3); \draw[densely dashed] (Z0)--(Q4);
        \draw[densely dashed] (Z1)--(Q1); \draw[densely dashed] (Z1)--(Q2); \draw[densely dashed] (Z1)--(Q4); \draw[densely dashed] (Z1)--(Q5);
        \draw[densely dashed] (Z2)--(Q0); \draw[densely dashed] (Z2)--(Q2); \draw[densely dashed] (Z2)--(Q3); \draw[densely dashed] (Z2)--(Q5);
      \end{scope}
      \foreach \i in {0,1,2}{
        \node[rectangle,draw,fill=blue!12,minimum width=16pt,minimum height=13pt,inner sep=0pt] at (X\i) {$x_\i$};
        \node[rectangle,draw,fill=red!12,minimum width=16pt,minimum height=13pt,inner sep=0pt] at (Z\i) {$z_\i$};
      }
      \foreach \i/\q in {0/A_0,1/A_1,2/A_2,3/B_0,4/B_1,5/B_2}
      \node[circle,draw,fill=gray!12,minimum size=13pt,inner sep=0pt] at (Q\i) {$\q$};
      \node[above] at (0,0.55) {$\tn{X}$ checks}; \node[above] at (4,0.75) {data qubits}; \node[above] at (8,0.55) {$\tn{Z}$ checks};
    \end{tikzpicture}%
  }
  \caption{The standard CSS Tanner graph of the $[[6,2,2]]$ code in Example~\ref{ex:622}. Solid and dashed edges are incident on $\tn{X}$- and $\tn{Z}$-check vertices, respectively.}
  \label{fig:tanner622-uncolored}
\end{figure}
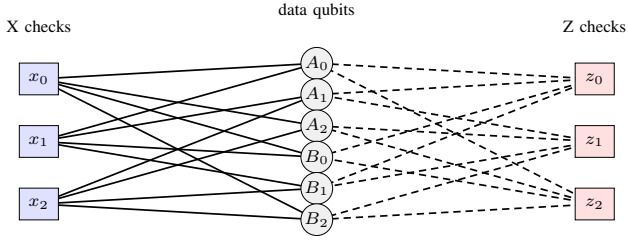

\subsection{Quantum Tanner Codes}
\label{sec:prelim-QT}

Quantum Tanner codes are defined by a pair of regular graphs fitting together in a certain way and a pair of classical codes of length equal to the degree of the graphs. We rephrase the original definition from \cite{LeverrierZemor22_1} to better suit our needs, and refer the reader to~\cite{MostadRosnesLin25_1, MostadRosnesLin25_2} for generalizations for which the results of \cref{sec:qTan-aut} concerning Tanner graph automorphisms hold or can be adapted to, respectively.

First, recall that the \textit{incidence graph} $\graph{{I}(\graph{G})}$ of a graph $\graph{G}=(\set{V}, \set{E})$ is the bipartite graph with vertices $\set{V}\sqcup\set{E}$ and edges $(v, e)$  for $v\in \set{V}, e\in \set{E}$ whenever $v\in e$. If the graph $\graph{G}$ has labeled local views, that is, for each vertex $v\in\set{V}$, the set of adjacent edges is ordered locally by a $\mu_v : \delta(v) \cong [1 : |\delta(v)|]$, then there is a natural labeling/coloring on the edges of $\graph{I(G)}$, namely, $\mu(v, e) = \mu_v(e)$.

We will define the Tanner code of a base graph and a local code through its Tanner graph. Let $\graph{G}^{\mat{H}}=(\set{V}_{\tn{c}}^\mat{H}\sqcup \set{V}_{\tn{d}}^\mat{H}, \set{E}^\mat{H})$ be a bipartite graph with $\ecard{\set{V}_\tn{d}^\mat{H}} = l$, and let $\graph{B}=(\set{V}_\tn{c}\sqcup\set{V}_\tn{d},\set{E})$ be an $(l, s)$-biregular graph with edges labeled/colored by $\mu\colon\set{E}\to \set{V}_{\tn{d}}^\mat{H}$ such that $\mu|_{\delta(v)}$ is an isomorphism for all $v\in \set{V}_{\tn{c}}$.
Define the graph
\begin{equation*}
  \graph{G} = \left((\set{V}_\tn{c}\times \set{V}_\tn{c}^\mat{H})\sqcup \set{V}_\tn{d}, \set{R}\right)
\end{equation*}
with edges
\begin{equation*}
  \set{R} = \{((v, t), q) : (v, q)\in \set{E} \tn{ and } (\mu(v, q), t) \in \set{E}^\mat{H} \}.
\end{equation*}  
Given orderings on the vertex sets, we can view $\graph{G}^\mat{H}$ as the Tanner graph of a classical code $\code C^\mat{H}$ with $|\set{V}_\tn{c}^{\mat{H}}| \times |\set{V}_\tn{d}^{\mat{H}}|$ parity-check matrix $\mat{H}$, and $\graph{G}$ as the Tanner graph of a classical code $\code{C}$. In that case, $\code C$ is called the \textit{Tanner code} with base graph $\graph{B}$ and local code $\code C^\mat{H}$, and we denote the parity-check matrix of $\code{C}$ by $\tn{Tan}(\graph{B}, \mat{H})$.


Let $\set{G}$ be a group with subsets $\set{A}, \set{B}$ satisfying $\inv{a}\in\set{A}$, $\inv{b}\in\set{B}$, and $agb \neq g$ for all $a \in \set{A}$, $g\in \set{G}$, and $b \in \set{B}$. 
We define $\set{X} = (\set{V}, \set{E}, \set{Q})$ as the following square complex. $\set{X}$ has vertices $\set{V} = \set{V}_0\sqcup\set{V}_1$ where $\set{V}_i = \set{G}\times \{i\}$, and edges $\set{E} = \set{E}_\tn{A}\sqcup\set{E}_\tn{B}$ where
\begin{IEEEeqnarray*}{rCl}
  \set{E}_\tn{A} & = &\{((g,0), (ag,1))\colon a \in \set{A}, g\in \set{G}\},
  \\
  \set{E}_\tn{B} & = &\{((g,0), (gb,1))\colon b \in \set{B}, g\in \set{G}\}.
\end{IEEEeqnarray*}
Equivalently, $(\set{V}, \set{E}_\tn{A})$ is the double cover of the left Cayley graph $\tn{Cay}_l(\set{G}, \set{A})$ and $(\set{V}, \set{E}_\tn{B})$ is the double cover of the right Cayley graph $\tn{Cay}_r(\set{G}, \set{B})$. The squares of $\set{X}$ are
\begin{equation*}
  \set{Q} = \{ ((g,0), (ag, 1), (gb, 1), (agb, 0))\colon a \in \set{A}, g\in \set{G}, b \in \set{B} \}.
\end{equation*}
From $\set{X}$ we define two graphs of (diagonals of) squares, $\graph{G}^\square_i = (\set{V}_i, \set{Q})$ for $i=0,1$, where squares are viewed as edges between the vertices on one of their diagonals.
Note that the square $q = ((g,0), (ag, 1), (gb, 1), (agb, 0))$ is labeled by $\set{A}\times \set{B}$ in a natural way in its four corners, namely
\begin{IEEEeqnarray*}{rClrCl}
  \mu_{(g,0)}(q)& = &(a,b),& \quad \mu_{(ag,1)}(q)& = &(\inv{a},b),
  \\
  \mu_{(gb,1)}(q)& = &(a, \inv{b}),& \quad \mu_{(agb,0)}(q)& = &(\inv{a}, \inv{b}).
\end{IEEEeqnarray*}

Let $\mat{H}_\tn{A}, \mat{H}_\tn{B}, \mat{G}_\tn{A}, \mat{G}_\tn{B}$ be parity-check and generator matrices for two classical codes
\begin{equation*}
  \code C_\tn{A} = \tn{ker}(\mat{H}_\tn{A}) = \tn{im}(\trans{\mat{G}}_\tn{A}),
  \quad
  \code C_\tn{B} = \tn{ker}(\mat{H}_\tn{B}) = \tn{im}(\trans{\mat{G}}_\tn{B})
\end{equation*}
of length $\ecard{\set{A}}$ and $\ecard{\set{B}}$, respectively. The quantum Tanner code on the above data is then defined as
\begin{equation*}
  \mat{H}_{\tn{X}} = \tn{Tan}(\graph{I}(\graph{G}^\square_0), \mat{G}_\tn{A} \otimes \mat{G}_\tn{B}), \ \ \
  \mat{H}_{\tn{Z}} = \tn{Tan}(\graph{I}(\graph{G}^\square_1), \mat{H}_\tn{A} \otimes \mat{H}_\tn{B}).
\end{equation*}

\section{Quantum Ordering Constraint}
\label{sec:quantum-ordering}

There are $m_{\tn{X}}$ $\tn{X}$-check ancillas and $m_{\tn{Z}}$ $\tn{Z}$-check ancillas, one for each measured stabilizer generator, together with $n$ data qubits.

We use the $\tn{C}\mat{X}$/$\tn{C}\mat{Z}$ formulation following ASC~\cite{Zhang-etal26_1sub}. An $\tn{X}$-check ancilla $a_i$ interacts with a data qubit $q$ through $(\tn{C}\mat{X})_{a_i\to q}$, while a $\tn{Z}$-check ancilla $b_j$ interacts with $q$ through $(\tn{C}\mat{Z})_{b_j,q}$. The arrow in $(\tn{C}\mat{X})_{a\to q}$ specifies that the $\tn{X}$-check ancilla $a$ controls the data qubit $q$. Since $\tn{C}\mat{Z}$ is symmetric,  its subscript lists the two qubits on which it acts. The $\tn{C}\mat{Z}$ interaction can be written as a $\tn{C}\mat{X}$ interaction conjugated by Hadamard gates on the data qubit:
\begin{equation*}
  (\tn{C}\mat{Z})_{b_j,q}=(\mat{I}_{b_j}\otimes\mat{H}_q)
  (\tn{C}\mat{X})_{b_j\to q}(\mat{I}_{b_j}\otimes\mat{H}_q).
\end{equation*}
Thus, replacing $\tn{C}\mat{X}$ by $\tn{C}\mat{Z}$ only changes the single-qubit basis of the data qubit before and after the interaction; it preserves the Tanner incidence and the two-qubit scheduling constraints. In the standard single-ancilla implementation, each Tanner edge corresponds to $1$ two-qubit interaction.

For a pair of checks $(x_i,z_j)$, we define the shared-support set
\begin{equation*}
  \set{S}_{ij}=\set{N}(x_i)\cap \set{N}(z_j).
\end{equation*}
This tick-independent set contains the data qubits adjacent to both checks. It follows from \eqref{eq:cssorth} that $\ecard{\set{S}_{ij}}$ is even.

\begin{lemma}[{$\tn{X}$--$\tn{Z}$ commutativity~\cite[Prop.~1]{Zhang-etal26_1sub}}]
  \label{lem:crossed-ordering}
  For distinct ancillas $a$ and $b$ coupled to a common data qubit $q$, the single-crossing identity is
  \begin{equation*}
    \bigl(\tn{C}\mat{Z}\bigr)_{b,q}\bigl(\tn{C}\mat{X}\bigr)_{a\to q}
    =\bigl(\tn{C}\mat{Z}\bigr)_{a,b}\bigl(\tn{C}\mat{X}\bigr)_{a\to q}\bigl(\tn{C}\mat{Z}\bigr)_{b,q}.
  \end{equation*}
  Consequently, reversing the relative order of one $\tn{X}$-type and one $\tn{Z}$-type interaction introduces an additional $\bigl(\tn{C}\mat{Z}\bigr)_{a,b}$.
\end{lemma}

We say an ASC schedule of depth $\lambda$ is a tick-assignment vector
\begin{equation*}
  \vect{t}=(t_e)_{e\in\set{E}(\graph{G})}\in [0:\lambda-1]^{\set{E}(\graph{G})},
\end{equation*}
where $t_e$ is the tick assigned to edge $e$. Given $\vect{t}$, for each overlapping check pair $(x_i,z_j)$, we define its inversion parity by
\begin{equation*}
  \eta_{ij}(\vect{t})=
  \sum_{q\in\set{S}_{ij}}
  \eI{t_{(x_i,q)}<t_{(z_j,q)}}\pmod 2.
\end{equation*}
Applying Lemma~\ref{lem:crossed-ordering} to all qubits in $\set{S}_{ij}$ gives the residual ancilla coupling $(\tn{C}\mat{Z})_{a_i, b_j}^{\eta_{ij}(\vect{t})}$. Since $(\tn{C}\mat{Z})^2=\mat{I}$, this coupling is absent if and only if $\eta_{ij}(\vect{t})=0$. This gives constraint (C3) below.

\begin{definition}[ASC Scheduling Constraints~{\cite[Sec.~III-D]{Zhang-etal26_1sub}}]
  \label{def:asc-constraints}
  The tick-assignment vector $\vect{t}$ is called ASC-feasible if it satisfies the following conditions:
  \begin{itemize}
  \item[(C1)] edges incident on any check vertex receive distinct ticks;
  \item[(C2)] edges incident on any data vertex receive distinct ticks;
  \item[(C3)] every overlapping $\tn{X}/\tn{Z}$ check pair satisfies
    \begin{equation}
      \eta_{ij}(\vect{t})=0,
      \label{eq:q-valid}
    \end{equation}
    i.e., the induced ordering has even inversion parity.
  \end{itemize}
\end{definition}

We now express these constraints directly in the language of the CSS Tanner graph $\graph{G}$. Following standard graph-theoretic terminology~\cite{Diestel17_1}, a \emph{matching} is a set of pairwise nonincident edges. A partition
\begin{equation*}
  \set{E}(\graph{G})=\set{M}_{0}\sqcup\cdots\sqcup\set{M}_{\lambda-1}
\end{equation*}
into matchings is equivalently a proper $\lambda$-edge coloring, with $1$ color assigned to each matching. For the CSS Tanner graph induced by any fixed CSS PCMs, conditions (C1) and (C2) state precisely that every tick class is a matching: no $2$ interactions sharing either a check qubit or a data qubit occur in the same layer. Thus, (C1) and (C2) are the classical edge-coloring constraints, whereas (C3) additionally orders the color classes so that every $\tn{X}/\tn{Z}$ overlap has even inversion parity.

\begin{definition}[Proper Ordered Edge Coloring]
  \label{def:proper-ordered-EC}
  A proper ordered $\lambda$-edge coloring is a map
  \begin{equation*}
    \varphi\colon\set{E}(\graph{G})\to [0:\lambda-1]
  \end{equation*}
  such that incident edges receive distinct colors. The ordering of the colors induces the gate order: $\varphi(e)<\varphi(e')$ means that the gate associated with $e$ is applied before the gate associated with $e'$.
\end{definition}

The following proposition follows directly from Definitions~\ref{def:asc-constraints} and~\ref{def:proper-ordered-EC}.
\begin{proposition}[Quantum-Constrained Edge-Coloring Criterion]
  \label{prop:q-valid_properEC}
  Following Definition~\ref{def:asc-constraints}, a tick-assignment vector $\vect{t}=(t_e)_{e\in\set{E}(\graph{G})}$ with $t_e\in [0:\lambda-1]$ is \emph{ASC-feasible} if and only if
  \vspace{-2ex}
  \begin{IEEEeqnarray}{rCl}
    \IEEEeqnarraymulticol{3}{l}{\tn{[Proper edge coloring]:}}
    \nonumber\\[-1mm]
    t_e & \neq & t_{e'}, \quad
    \forall\,v\in\set{V}(\graph{G}),\ e\neq e'\in\delta(v),
    \label{eq:properEC_ASC}
    \\[1mm]
    \IEEEeqnarraymulticol{3}{l}{\tn{[Quantum-ordering constraint]:}}
    \nonumber\\[-1mm]
    \eta_{ij}(\vect{t}) & = & 0 \pmod 2, \quad
    \forall\,(i,j)\tn{ with }\set{S}_{ij}\neq\emptyset.
    \label{eq:q-ordering_ASC}
  \end{IEEEeqnarray}
  The first condition is precisely a proper edge coloring of the Tanner graph, with ticks playing the role of colors, while the second imposes the additional ordering-parity condition required for a valid syndrome-extraction schedule.
\end{proposition}

\begin{definition}[Optimal Syndrome-Extraction Circuit Depth]
  \label{def:sec-depth}
  For a fixed pair of CSS PCMs $(\mat{H}_{\tn{X}},\mat{H}_{\tn{Z}})$, the optimal syndrome-extraction circuit depth in the standard single-ancilla model is defined as
  \begin{IEEEeqnarray*}{rCl}
    \lambda_{\tn{SEC}}^{\star}(\mat{H}_{\tn{X}},\mat{H}_{\tn{Z}})
    & \eqdef & \min\{\lambda:\tn{an ASC-feasible schedule}
    \\
    &&\qquad\qquad\tn{of depth $\lambda$ exists}\}.
  \end{IEEEeqnarray*}
\end{definition}

Proposition~\ref{prop:q-valid_properEC} concerns the correctness of syndrome measurement and does not imply distance preservation or hook-error optimality. In particular, different $\tn{C}\mat{X}$ orderings of the same depth may lead to different residual errors and circuit distances~\cite{StrikisBrowneBeverland26_1sub,Bravyi-etal24_1}.

\begin{figure*}[!t]
  \begin{minipage}[t]{0.485\textwidth}
    \centering
    \resizebox{0.95\columnwidth}{!}{%
      \begin{tikzpicture}[every node/.style={font=\scriptsize}]
        \coordinate (X0) at (0,0); \coordinate (X1) at (0,-0.9); \coordinate (X2) at (0,-1.8);
        \coordinate (Z0) at (8,0); \coordinate (Z1) at (8,-0.9); \coordinate (Z2) at (8,-1.8);
        \coordinate (Q0) at (4,0.225); \coordinate (Q1) at (4,-0.225); \coordinate (Q2) at (4,-0.675);
        \coordinate (Q3) at (4,-1.125); \coordinate (Q4) at (4,-1.575); \coordinate (Q5) at (4,-2.025);
        \begin{scope}[line width=0.7pt]
          \draw[red!75!black] (X0)--(Q0); \draw[blue!70!black] (X0)--(Q2); \draw[orange!85!black] (X0)--(Q3); \draw[green!55!black] (X0)--(Q5);
          \draw[blue!70!black] (X1)--(Q0); \draw[red!75!black] (X1)--(Q1); \draw[green!55!black] (X1)--(Q3); \draw[orange!85!black] (X1)--(Q4);
          \draw[blue!70!black] (X2)--(Q1); \draw[red!75!black] (X2)--(Q2); \draw[green!55!black] (X2)--(Q4); \draw[orange!85!black] (X2)--(Q5);
          \draw[green!55!black,densely dashed] (Z0)--(Q0); \draw[orange!85!black,densely dashed] (Z0)--(Q1); \draw[blue!70!black,densely dashed] (Z0)--(Q3); \draw[red!75!black,densely dashed] (Z0)--(Q4);
          \draw[green!55!black,densely dashed] (Z1)--(Q1); \draw[orange!85!black,densely dashed] (Z1)--(Q2); \draw[blue!70!black,densely dashed] (Z1)--(Q4); \draw[red!75!black,densely dashed] (Z1)--(Q5);
          \draw[orange!85!black,densely dashed] (Z2)--(Q0); \draw[green!55!black,densely dashed] (Z2)--(Q2); \draw[red!75!black,densely dashed] (Z2)--(Q3); \draw[blue!70!black,densely dashed] (Z2)--(Q5);
        \end{scope}
        \foreach \i in {0,1,2}{
          \node[rectangle,draw,fill=blue!12,minimum width=16pt,minimum height=13pt,inner sep=0pt] at (X\i) {$x_\i$};
          \node[rectangle,draw,fill=red!12,minimum width=16pt,minimum height=13pt,inner sep=0pt] at (Z\i) {$z_\i$};
        }
        \foreach \i/\q in {0/A_0,1/A_1,2/A_2,3/B_0,4/B_1,5/B_2}
        \node[circle,draw,fill=gray!12,minimum size=13pt,inner sep=0pt] at (Q\i) {$\q$};
        \node[above] at (0,0.55) {$\tn{X}$ checks}; \node[above] at (4,0.75) {data qubits}; \node[above] at (8,0.55) {$\tn{Z}$ checks};
        \foreach \xx/\cc/\lab in {0.5/red!75!black/0,2.3/orange!85!black/1,4.1/green!55!black/2,5.9/blue!70!black/3}{
          \draw[\cc,line width=0.9pt] (\xx,-3.1)--++(0.55,0); \node[anchor=west] at (\xx+0.65,-3.1) {tick \lab};
        }
      \end{tikzpicture}%
    }
    \caption{A quantum-valid $4$-edge coloring of the CSS Tanner graph in Fig.~\ref{fig:tanner622-uncolored}. Solid and dashed edges represent $\tn{C}\mat{X}$ and $\tn{C}\mat{Z}$ gates, respectively. Each color class is a perfect matching of size $6$, and $\eta_{ij}(\vect{t})=0$ for every overlapping $\tn{X}/\tn{Z}$ check pair.}
    \label{fig:tanner622}
  \end{minipage}\hfill
  \begin{minipage}[t]{0.485\textwidth}
    \centering
    \resizebox{0.95\columnwidth}{!}{%
      \begin{tikzpicture}[every node/.style={font=\scriptsize}]
        \coordinate (X0) at (0,0); \coordinate (X1) at (0,-0.9); \coordinate (X2) at (0,-1.8);
        \coordinate (Z0) at (8,0); \coordinate (Z1) at (8,-0.9); \coordinate (Z2) at (8,-1.8);
        \coordinate (Q0) at (4,0.225); \coordinate (Q1) at (4,-0.225); \coordinate (Q2) at (4,-0.675);
        \coordinate (Q3) at (4,-1.125); \coordinate (Q4) at (4,-1.575); \coordinate (Q5) at (4,-2.025);
        \draw[blue!70!black!25,line width=0.7pt] (X1)--(Q0); \draw[red!75!black!25,line width=0.7pt] (X1)--(Q1);
        \draw[green!55!black!25,line width=0.7pt] (X1)--(Q3); \draw[orange!85!black!25,line width=0.7pt] (X1)--(Q4);
        \draw[blue!70!black!25,line width=0.7pt] (X2)--(Q1); \draw[red!75!black!25,line width=0.7pt] (X2)--(Q2);
        \draw[green!55!black!25,line width=0.7pt] (X2)--(Q4); \draw[orange!85!black!25,line width=0.7pt] (X2)--(Q5);
        \draw[green!55!black!25,line width=0.7pt,densely dashed] (Z1)--(Q1); \draw[orange!85!black!25,line width=0.7pt,densely dashed] (Z1)--(Q2);
        \draw[blue!70!black!25,line width=0.7pt,densely dashed] (Z1)--(Q4); \draw[red!75!black!25,line width=0.7pt,densely dashed] (Z1)--(Q5);
        \draw[orange!85!black!25,line width=0.7pt,densely dashed] (Z2)--(Q0); \draw[green!55!black!25,line width=0.7pt,densely dashed] (Z2)--(Q2);
        \draw[red!75!black!25,line width=0.7pt,densely dashed] (Z2)--(Q3); \draw[blue!70!black!25,line width=0.7pt,densely dashed] (Z2)--(Q5);
        \draw[red!75!black,line width=0.7pt] (X0)--(Q0); \draw[blue!70!black,line width=0.7pt] (X0)--(Q2);
        \draw[orange!85!black,line width=0.7pt] (X0)--(Q3); \draw[green!55!black,line width=0.7pt] (X0)--(Q5);
        \draw[green!55!black,line width=0.7pt,densely dashed] (Z0)--(Q0); \draw[orange!85!black,line width=0.7pt,densely dashed] (Z0)--(Q1);
        \draw[blue!70!black,line width=0.7pt,densely dashed] (Z0)--(Q3); \draw[red!75!black,line width=0.7pt,densely dashed] (Z0)--(Q4);
        \foreach \i in {0,1,2}{
          \node[rectangle,draw,fill=blue!12,minimum width=16pt,minimum height=13pt,inner sep=0pt] at (X\i) {$x_\i$};
          \node[rectangle,draw,fill=red!12,minimum width=16pt,minimum height=13pt,inner sep=0pt] at (Z\i) {$z_\i$};
        }
        \foreach \i/\q in {0/A_0,1/A_1,2/A_2,3/B_0,4/B_1,5/B_2}
        \node[circle,draw,fill=gray!12,minimum size=13pt,inner sep=0pt] at (Q\i) {$\q$};
        \node[above] at (0,0.55) {$\tn{X}$ checks}; \node[above] at (4,0.75) {data qubits}; \node[above] at (8,0.55) {$\tn{Z}$ checks};
        \foreach \xx/\cc/\lab in {0.5/red!75!black/0,2.3/orange!85!black/1,4.1/green!55!black/2,5.9/blue!70!black/3}{
          \draw[\cc,line width=0.9pt] (\xx,-3.1)--++(0.55,0); \node[anchor=west] at (\xx+0.65,-3.1) {tick \lab};
        }
      \end{tikzpicture}%
    }
    \caption{The order-$3$ LocalASC reduction for the $[[6,2,2]]$ code. The $8$ saturated edges are orbit representatives and correspond to the tick variables of the orbit model. The remaining $16$ edges are shown faintly, since all edges in the same orbit receive the same tick. Each monochromatic orbit contains $3$ edges.}
    \label{fig:LocalASC-orbits622}
  \end{minipage}
\end{figure*}

\subsection{Quantum-Constrained Edge Chromatic Number}
\label{sec:qcec-number}

We now define the quantum-constrained edge chromatic number using the scheduling constraints in Definition~\ref{def:asc-constraints}.
\begin{definition}[Quantum-Constrained Edge Chromatic Number]
  \label{def:qcec}
  For a fixed pair of CSS PCMs $(\mat{H}_{\tn{X}},\mat{H}_{\tn{Z}})$, the quantum-constrained edge chromatic number is
  \begin{IEEEeqnarray*}{rCl}
    \Qchi(\mat{H}_{\tn{X}},\mat{H}_{\tn{Z}})
    & = & \min\{\lambda:\tn{a proper ordered $\lambda$-edge coloring}
    \\
    && \qquad\qquad\tn{satisfying~\eqref{eq:q-ordering_ASC} exists}\}.
  \end{IEEEeqnarray*}
\end{definition}

\begin{theorem}[Optimal Two-Qubit-Depth Characterization]
  \label{thm:qcec-depth}
  In the standard single-ancilla $\tn{C}\mat{X}$/$\tn{C}\mat{Z}$ syndrome-measurement model,
  \begin{equation*}
    \lambda_{\tn{SEC}}^{\star}(\mat{H}_{\tn{X}},\mat{H}_{\tn{Z}})=\Qchi(\mat{H}_{\tn{X}},\mat{H}_{\tn{Z}}).
  \end{equation*}
\end{theorem}

\begin{IEEEproof}
  Each two-qubit layer defines $1$ tick, and conditions (C1)--(C3) characterize the valid schedules by Proposition~\ref{prop:q-valid_properEC}. Minimizing the number of ticks is therefore equivalent to minimizing the number of colors in Definition~\ref{def:qcec}.
\end{IEEEproof}

The Tanner graph $\graph{G}$ is bipartite. Hence, if (C3) is omitted, K\"onig's line-coloring theorem gives $\chi'(\graph{G})=\Delta(\graph{G})$~\cite{Diestel17_1}. Since the quantum constraint (C3) is imposed in addition to proper edge coloring, clearly we have
\begin{equation*}
\Qchi(\mat{H}_{\tn{X}},\mat{H}_{\tn{Z}})\ge \chi'(\graph{G})=\Delta(\graph{G}).
\end{equation*}
Thus, any depth exceeding $\Delta(\graph{G})$ is due to the quantum-ordering constraint. We define the \emph{quantum scheduling gap} as
\begin{equation*}
  \Qgap(\mat{H}_{\tn{X}},\mat{H}_{\tn{Z}})\eqdef\Qchi(\mat{H}_{\tn{X}},\mat{H}_{\tn{Z}})-\Delta(\graph{G})\ge0.
\end{equation*}
For a $\Delta$-regular bipartite Tanner graph $\graph{G}$, every $\Delta$-edge coloring is a decomposition into $\Delta$ perfect matchings. Hence, when $\Qgap=0$, the optimal schedule consists of $\Delta(\graph{G})$ perfect matchings whose ordering satisfies all quantum-ordering constraints in~\eqref{eq:q-valid}.

\begin{example}[Example~\ref{ex:622} Continued]
  \label{ex:622-schedule}
  Fig.~\ref{fig:tanner622} illustrates a valid $4$-tick edge coloring of the Tanner graph in Fig.~\ref{fig:tanner622-uncolored} for the code in Example~\ref{ex:622}. The graph has $24$ edges, $1$ for each two-qubit interaction in a syndrome-extraction round; each check has weight $4$, and each data qubit has degree $4$. Thus, $w=\Delta=4$. Its $4$ tick classes form a partition $\set{E}(\graph{G})=\set{M}_0\sqcup\cdots\sqcup\set{M}_{3}$, where each $\set{M}_{\ell}$ is a perfect matching containing $6$ edges, one incident on each data-qubit vertex and each check vertex. Hence, (C1) and (C2) hold. For (C3), direct evaluation of the $9$ overlapping check pairs gives
  \begin{equation*}
    \bigl(\eta_{ij}(\vect{t})\bigr)_{0\le i,j\le 2}
    =
    \begin{pmatrix}
      2&0&2
      \\
      2&2&0
      \\
      0&2&2
    \end{pmatrix}
    \bmod 2
    =\mat{0}_{3\times3}.
  \end{equation*}
  For instance, the $4$ shared qubits $(A_0,A_2,B_0,B_2)$ of $(x_0,z_2)$ contribute
  \begin{IEEEeqnarray*}{rCl}
    \eta_{0,2}(\vect{t})& = &\bigl(\eI{0<1}+\eI{3<2}+\eI{1<0}+\eI{2<3}\bigr)
    \\
    & = &2\equiv0\pmod 2.
  \end{IEEEeqnarray*}
  Thus, the assignment satisfies (C3) and is ASC-feasible by Definition~\ref{def:asc-constraints}. Equivalently, Proposition~\ref{prop:q-valid_properEC} implies that no residual $\tn{C}\mat{Z}$ coupling remains between any overlapping ancilla pair. Since $w=\Delta(\graph{G})=4$, the scheduling lower bound is attained, and
  \begin{equation*}
    \Qchi=4=\Delta.
  \end{equation*}
  Thus, the pair of CSS PCMs is weight-optimal, the displayed schedule is lower-bound-saturating and hence depth-optimal, and $\Qgap=0$ for the $[[6,2,2]]$ running example.
\end{example}

Let $\Delta_{\tn{D}}$ denote the maximum data-qubit degree of $\graph{G}$. Every ASC-feasible schedule satisfies
\begin{equation*}
  \lambda\ge\max\{w,\Delta_{\tn{D}}\}=\Delta(\graph{G}).
\end{equation*}
We call an ASC-feasible schedule \emph{depth-optimal} if its depth is $\Qchi$, and \emph{lower-bound-saturating} if its depth is $\Delta(\graph{G})$. Every lower-bound-saturating schedule is depth-optimal. We call a pair of measured CSS PCMs \emph{weight-optimal} if
\begin{equation*}
  \Qchi(\mat{H}_{\tn{X}},\mat{H}_{\tn{Z}})=w.
\end{equation*}
Since $\Qchi\ge\Delta(\graph{G})=\max\{w,\Delta_{\tn{D}}\}$, weight optimality implies $w\ge\Delta_{\tn{D}}$ and the existence of a depth-$w$ schedule attaining the degree lower bound.

\section{LocalASC: Local Automorphism-Aware Syndrome Compilation}
\label{sec:LocalASC}

In this section, we introduce LocalASC with a motivating example that illustrates how a code's algebraic structure can be exploited to perform ASC locally.

\begin{example}[Example~\ref{ex:622-schedule} Continued]
  \label{ex:LocalASC-example}
  
  For the code in Example~\ref{ex:622}, the unreduced ASC model uses $24$ edge variables and finds a depth-optimal $4$-tick schedule. The group $\Aut(\graph{G})$ has order $48$ and contains elements of orders $1$, $2$, $3$, $4$, and $6$.
  
  Let $\mat{P}$ be the cyclic permutation matrix in Example~\ref{ex:622}, corresponding to the translation $0\mapsto 1\mapsto 2\mapsto 0$ on $\set{C}_3$. Since $\mat{A}=\mat{B}=\mat{I}+\mat{P}$ is circulant, $\mat{P}\mat{A}\trans{\mat{P}}=\mat{A}$ and $\mat{P}\mat{B}\trans{\mat{P}}=\mat{B}$. Define the permutation of the $2$ data-qubit blocks by
  \begin{equation*}
    \mat{Q}=
    \begin{pmatrix}
      \mat{P}&\mat{0}\\
      \mat{0}&\mat{P}
    \end{pmatrix}.
  \end{equation*}
  Then, $\mat{P}\mat{H}_{\tn{X}}\trans{\mat{Q}}=\mat{H}_{\tn{X}}$ and $\mat{P}\mat{H}_{\tn{Z}}\trans{\mat{Q}}=\mat{H}_{\tn{Z}}$. Thus, the simultaneous cyclic permutation
  \begin{equation*}
    h=(x_0x_1x_2)(z_0z_1z_2)(A_0A_1A_2)(B_0B_1B_2)
  \end{equation*}
  preserves both the $\tn{X}$- and $\tn{Z}$-type Tanner incidences. Hence, $h\in\Aut(\graph{G})$. Moreover, $h^3=\operatorname{id}$ and $h\neq\operatorname{id}$, so
  \begin{equation*}
    \egen{h}\cong\set{C}_3\le\Aut(\graph{G}).
  \end{equation*}
  Therefore, $\egen{h}$ is precisely the cyclic translation symmetry inherited from the two-block construction.
  
  The numerical search identifies $\egen{h}$ as supporting an ASC-feasible schedule at $\lambda=4$. The two nonidentity elements of $\egen{h}$ move every Tanner edge. Hence, each edge orbit contains exactly three edges, and
  \begin{equation*}
    \bigcard{\quotient{\set{E}(\graph{G})}{\egen{h}}}=\frac{\ecard{\set{E}(\graph{G})}}{\ecard{\egen{h}}}
    =\frac{24}{3}=8.
  \end{equation*}
  Fig.~\ref{fig:LocalASC-orbits622} shows the corresponding orbit representation. Solving for one tick variable per orbit and lifting the resulting assignment recovers the coloring in Fig.~\ref{fig:tanner622}. Each of the $8$ orbit variables determines the common tick assigned to the $3$ Tanner edges in its orbit. Each color class in Fig.~\ref{fig:tanner622} is a perfect matching, and Example~\ref{ex:622-schedule} verifies (C3) for all $9$ overlapping check pairs. Therefore, the lifted coloring is quantum-valid.
\end{example}

Motivated by Example~\ref{ex:LocalASC-example}, we now formalize the orbit reduction and the LocalASC procedure. We consider a subgroup $\set{H}\le\Aut(\graph{G})$ whose elements preserve the vertex classes $\set{V}_{\tn{X}}$, $\set{V}_{\tn{D}}$, and $\set{V}_{\tn{Z}}$.

\subsubsection*{Orbit Model} We say that a global schedule $\vect{t}$ is \emph{$\set{H}$-invariant} if
\begin{IEEEeqnarray*}{c}
  t_{h(e)}=t_e,\qquad \forall\,h\in\set{H},\ \forall\,e\in\set{E}(\graph{G}),
\end{IEEEeqnarray*}
or, equivalently, if $\vect{t}$ is constant on every edge orbit.

\begin{definition}[Orbit-Reduced ASC Model]
  \label{def:orbit-reduced-ASC}
  Let $\set{H}$ be a subgroup of $\Aut(\graph{G})$ acting on $\set{E}(\graph{G})$. Let $\set{O}\eqdef\quotient{\set{E}(\graph{G})}{\set{H}}$ denote the set of edge orbits, and let $\pi\colon\set{E}(\graph{G})\to\set{O}$ be the orbit projection defined by $\pi(e)\eqdef\set{H}e$. For a fixed depth $\lambda$, introduce one tick variable $\tau_o\in[0:\lambda-1]$ for each orbit $o\in\set{O}$. The corresponding \emph{lifted} assignment on the original Tanner edges is
  \begin{IEEEeqnarray*}{rCl}
    \set{E}(\graph{G}) & \xrightarrow{\ \pi\ } & \set{O}\xrightarrow{\ \tau\ }[0:\lambda-1],\\
    e & \longmapsto & \pi(e)\longmapsto\tau_{\pi(e)}\eqdef t_e.
  \end{IEEEeqnarray*}
  Hence, the lifted schedule is $\set{H}$-invariant by construction. The $\set{H}$-orbit model is obtained from Definition~\ref{def:asc-constraints} by the exact substitution $t_e=\tau_{\pi(e)}$ for every $e\in\set{E}(\graph{G})$:
  \begin{IEEEeqnarray}{rCl}
    \IEEEeqnarraymulticol{3}{l}{\tn{[Proper edge coloring]:}\quad \forall\,v\in\set{V}(\graph{G}),\ \forall\,e\neq e'\in\delta(v),}\nonumber\\[-1mm]
    \tau_{\pi(e)} & \neq & \tau_{\pi(e')},
    \label{eq:properEC_orbit}
    \\[1mm]
    \IEEEeqnarraymulticol{3}{l}{\tn{[Quantum-ordering constraint]:}\quad \forall\,(i,j)\tn{ with }\set{S}_{ij}\neq\emptyset,}
    \nonumber\\[1mm]
    \sum_{q\in\set{S}_{ij}}\eI{\tau_{\pi(x_i,q)}<\tau_{\pi(z_j,q)}} & = & 0 \pmod 2. \label{eq:q-ordering_orbit}
  \end{IEEEeqnarray}
  All vertices, incident-edge sets $\delta(v)$, and common supports $\set{S}_{ij}$ appearing in these constraints are those of the original Tanner graph $\graph{G}$.
\end{definition}

\begin{theorem}[Orbit-Invariant Schedule Lifting]
  \label{thm:orbit-invariant-schedule-lifting}
  Fix $\lambda\ge1$ and a subgroup $\set{H}\le\Aut(\graph{G})$, and let $\pi$ be the orbit map in Definition~\ref{def:orbit-reduced-ASC}. For an orbit assignment $\tau=(\tau_o)_{o\in\set{O}}$, define
  \begin{equation*}
    \vect{t}=\bigl(\tau_{\pi(e)}\bigr)_{e\in\set{E}(\graph{G})}.
  \end{equation*}
  Then $\tau$ satisfies Definition~\ref{def:orbit-reduced-ASC} if and only if $\vect{t}$ satisfies Definition~\ref{def:asc-constraints} and is $\set{H}$-invariant. Consequently, this correspondence is a bijection between the two sets of assignments in the same tick domain $[0:\lambda-1]$.
\end{theorem}
\begin{IEEEproof}
  See Appendix~\ref{app:proof-orbit-lift}.
\end{IEEEproof}

By Theorem~\ref{thm:orbit-invariant-schedule-lifting}, feasibility of the orbit model at $\lambda=\max(w,\Delta_{\tn{D}})$ certifies a lower-bound-saturating, and hence depth-optimal, schedule. In contrast, infeasibility excludes only $\set{H}$-invariant schedules and does not imply global infeasibility. Let $\lambda_{\tn{OR}}(\set{H})$ denote the minimum depth among $\set{H}$-invariant ASC-feasible schedules. If $\set{H}_2\leq\set{H}_1$, then
\begin{IEEEeqnarray*}{rCl}
  \Qchi(\mat{H}_{\tn{X}},\mat{H}_{\tn{Z}})=\lambda_{\tn{SEC}}^\star(\mat{H}_{\tn{X}},\mat{H}_{\tn{Z}})& = &\lambda_{\tn{OR}}(\{1\})
  \\
  & \leq &\lambda_{\tn{OR}}(\set{H}_2)\leq\lambda_{\tn{OR}}(\set{H}_1).
\end{IEEEeqnarray*}
Thus, a smaller subgroup imposes fewer equality constraints on the edge ticks and cannot increase the minimum depth.

\subsubsection*{LocalASC Procedure} Algorithm~\ref{alg:LocalASC} summarizes the complete procedure. In practice, if two distinct edges incident on the same check or data vertex belong to the same orbit, they must receive the same tick and violate the proper-edge-coloring constraint. We exclude such subgroups before solving the orbit-reduced model. A feasible orbit assignment is accepted only after its lift has been verified against all global constraints (C1)--(C3).

\begin{algorithm}[t!]
  \caption{LocalASC at a fixed depth}
  \label{alg:LocalASC}
  \KwData{A fixed pair of measured CSS PCMs $\mat{H}_{\tn{X}},\mat{H}_{\tn{Z}}$, a depth $\lambda\ge\max(w,\Delta_{\tn{D}})$, and an ordered family of subgroups $\set{H}_1,\ldots,\set{H}_{\Gamma}\le\Aut(\graph{G})$.}
  \KwResult{An ASC-feasible global schedule $\vect{t}$, or no schedule found in the selected invariant classes.}
  Construct the CSS Tanner graph $\graph{G}$\;
  Form all constraints of Definition~\ref{def:asc-constraints} on $\graph{G}$\;
  \For{$j\in[1:\Gamma]$}{
    $\set{H}\leftarrow\set{H}_j$\;
    Form the edge-orbit set $\set{O}$ and orbit map $\pi$ of Definition~\ref{def:orbit-reduced-ASC}\;
    \If{two distinct edges incident on the same vertex belong to the same orbit}{
      \Continue\;
    }
    Introduce $\tau_o\in[0:\lambda-1]$ for each $o\in\set{O}$\;
    Replace $t_e$ by $\tau_{\pi(e)}$ in every original constraint\;
    Solve the resulting orbit-reduced model at depth $\lambda$ within its time limit\;
    \If{a feasible orbit assignment $\vect{\tau}$ is found}{
      Set $t_e\leftarrow\tau_{\pi(e)}$ for every $e\in\set{E}(\graph{G})$\;
      Independently verify \eqref{eq:properEC_ASC} and \eqref{eq:q-ordering_ASC} on the lifted schedule\;
      \If{both constraints hold on the full Tanner graph $\graph{G}$}{
        \Return $\vect{t}$\;
      }
    }
  }
  \Return no ASC-feasible schedule found at depth $\lambda$ among the selected subgroups\;  
\end{algorithm}

\subsubsection*{Choice of the Subgroup Family $\set{H}_1,\ldots,\set{H}_{\Gamma}$} For a general CSS code, we compute $\Aut(\graph{G})$ using \texttt{nauty}~\cite{McKayPiperno14_1}, with \texttt{VF2}~\cite{CordellaFoggiaSansoneVento04_1} when necessary. If $\Aut(\graph{G})$ contains a nontrivial element of odd order, we take $\set{H}_1$ to be generated by one of maximum odd order. Although $\set{H}_1$ contains no automorphism of order $2$, distinct edges incident on the same vertex may still belong to the same orbit. Algorithm~\ref{alg:LocalASC} excludes such subgroups before solving. If $\set{H}_1$ passes this check, we first solve the model for $\set{H}_1$ with a $10$ seconds time limit. If $\Aut(\graph{G})$ has no nontrivial element of odd order, or if $\set{H}_1$ places two edges incident on the same vertex in the same orbit or its orbit-reduced model is infeasible at depth $\lambda$, we try subgroups of $\Aut(\graph{G})$ in decreasing order of their cardinalities. These include cyclic subgroups, subgroups generated by $2$ elements, and $\Aut(\graph{G})$ itself, and each model is solved with a $60$ seconds time limit. If the first attempt to solve the model for $\set{H}_1$ reaches its time limit without an outcome, we first try the selected subgroups of larger cardinality than $\set{H}_1$ and then solve the model for $\set{H}_1$ again with a $30$ minutes time limit. Algorithm~\ref{alg:LocalASC} returns the first ASC-feasible schedule found among the selected subgroups. The family need not contain every subgroup of $\Aut(\graph{G})$, since the nontrivial-subgroup models are used only to search for a feasible schedule with fewer variables than the unreduced model. The identity subgroup $\set{H}=\{1\}$ recovers the unreduced ASC model.

\section{Constructing Automorphisms for LocalASC}
\label{sec:constructing-automorphisms}

LocalASC requires a verified subgroup $\set{H}\leq\Aut(\graph{G})$ of the Tanner graph associated with the measured CSS PCMs. This subgroup does not need to be the full Tanner graph automorphism group. When the code construction provides candidate symmetries, we construct the induced permutations of the checks and data qubits and verify them directly on the measured PCMs. If no construction-based candidate is available or all candidates fail verification, we use the generic graph-automorphism routine described in Section~\ref{sec:LocalASC}. In this section, we describe how to construct such subgroups, especially for two-block and quantum Tanner codes.

\subsection{Automorphism Subgroups for Two-Block CSS Codes}
\label{sec:AutG_TwoBlock-codes}

For a two-block CSS code of length $n$, the square matrices $\mat{A}$ and $\mat{B}$ have order $\nu\eqdef\nicefrac{n}{2}$, so both PCMs have size $\nu\times2\nu$ as in~\eqref{eq:Hz-Hx_two-block-codes}. The first and second blocks of $\nu$ columns correspond to the data-qubit sets $\{q_1,\ldots,q_\nu\}$ and $\{q_{\nu+1},\ldots,q_{2\nu}\}$, respectively.

For the construction-based candidates, identify the $\tn{X}$-check indices, the $\tn{Z}$-check indices, and the indices within each data-qubit block with the elements of the following abelian group
\begin{equation*}
  \set{G}=\set{C}_{d_1}\times\cdots\times\set{C}_{d_r},\quad\nu=\ecard{\set{G}}=\prod_{j=1}^{r}d_j.
\end{equation*}
For $\vect{u}=(u_1,\ldots,u_r)$ and $\vect{g}=(g_1,\ldots,g_r)$ in $\set{G}$, define
\begin{equation}
  \rho_{\vect{g}}(\vect{u})\eqdef\vect{u}+\vect{g}.
  \label{eq:translation-index}
\end{equation}
In addition, in each coordinate $j$, the addition is modulo $d_j$. We apply $\rho_{\vect{g}}$ to the $\tn{X}$-check indices, the $\tn{Z}$-check indices, and the indices within each data-qubit block. We verify the induced permutations on the measured CSS PCMs.

\begin{enumerate}
\item \emph{Cyclic coordinates.} For $r=1$ and $d_1=\nu$, translation by $\vect{g}=(1)$ applies the same cyclic shift to the $\tn{X}$-check indices, the $\tn{Z}$-check indices, and the qubit indices within each data-qubit block. We use this permutation only if
  \begin{equation}
    \mat{P}\mat{H}_{\tn{X}}\trans{\mat{Q}}=\mat{H}_{\tn{X}},
    \qquad
    \mat{P}\mat{H}_{\tn{Z}}\trans{\mat{Q}}=\mat{H}_{\tn{Z}},
    \label{eq:permutation-identities_two-block-codes}
  \end{equation}
  where $\mat{P}$ applies the shift to the $\tn{X}$- and $\tn{Z}$-check indices and $\mat{Q}$ applies it within each data-qubit block. If both identities hold, the pairs $(\mat{P}^l,\mat{Q}^l)$ for $l=0,\ldots,\nu-1$ form the translation subgroup $\set{C}_{\nu}$.
  
\item \emph{Noncyclic abelian coordinates.} For a noncyclic abelian candidate, take $r=2$ or $3$ with $d_j\geq2$. For each coordinate $j$, choose $\vect{g}$ with $g_j=1$ and $g_\ell=0$ for every $\ell\neq j$. The permutations $\mat{P}$ and $\mat{Q}$ induced by each $\rho_{\vect{g}}$ must satisfy both identities in~\eqref{eq:permutation-identities_two-block-codes}. The verified translations generate an abelian subgroup.
\end{enumerate}

Here, we use the standard translation symmetry of two-block group-algebra codes~\cite{LinPryadko24_1, AydinTamoBarg26_1sub}. Let $a,b\in\Field_2[\set{G}]$, and let $\mat{A}=\mat{L}(a)$ and $\mat{B}=\mat{R}(b)$ be their matrices under left and right multiplication on $\set{G}$, respectively. When $\set{G}$ is abelian, simultaneous translation of the checks and both data-qubit blocks by any $\vect{g}\in\set{G}$ preserves the two PCMs in~\eqref{eq:Hz-Hx_two-block-codes}. Once the coordinate translations have been verified, their compositions give all translations in $\set{G}$. The implementation used here constructs the full translation group, corresponding to the coset-based code $\set{Q}_{\set{G}}^{\set{K}}(a,b)$ with $\set{K}=\{1\}$. It does not implement the general coset space $\quotient{\set{G}}{\set{K}}$ for an arbitrary subgroup $\set{K}\le\set{G}$. The algebraic structure of the general coset-based two-block construction may also yield nonabelian candidates for $\set{H}$, which we do not study here.

\subsubsection*{Construction of the Translation Action} Algorithm~\ref{alg:constructH_abelian-2BGA} forms the translation generators using~\eqref{eq:translation-index} and verifies their induced check and data-qubit permutations on the measured CSS PCMs.

\begin{algorithm}[t!]
  \caption{Construction and verification of an automorphism subgroup}
  \label{alg:constructH_abelian-2BGA}
  \KwData{Binary CSS PCMs $\mat{H}_{\tn{X}},\mat{H}_{\tn{Z}}$ with the same data-qubit columns.}
  \KwResult{A verified subgroup $\set{S}\le\Aut(\graph{G})$, represented by a set of generators $\set{T}$, or no construction-based subgroup found.}
  \If{both PCMs have size $\nu\times2\nu$}{
    Form $\set{T}$ from the shift by $\vect{g}=(1)$ in $\set{C}_{\nu}$ using~\eqref{eq:translation-index}, acting on the $\tn{X}$- and $\tn{Z}$-check indices and within each data-qubit block\;
    \If{the induced $\mat{P}$ and $\mat{Q}$ satisfy both identities in~\eqref{eq:permutation-identities_two-block-codes}}{
      \Return $\set{S}=\egen{\set{T}}$\;
    }
    \For{each $\set{G}=\set{C}_{d_1}\times\cdots\times\set{C}_{d_r}$ with $r=2$ or $3$, $d_j\geq2$, and $\prod_{j=1}^{r}d_j=\nu$}{
      Form $\set{T}$ from the $r$ coordinate shifts $\rho_{\vect{g}}$ with $g_j=1$ and $g_\ell=0$ for $\ell\neq j$, acting on the $\tn{X}$- and $\tn{Z}$-check indices and within each data-qubit block\;
      \If{each generator induces $\mat{P}$ and $\mat{Q}$ satisfying both identities in~\eqref{eq:permutation-identities_two-block-codes}}{
        \Return $\set{S}=\egen{\set{T}}$\;
      }
    }
  }
  \Return no construction-based subgroup found\;
\end{algorithm}

The output $\set{S}$ of Algorithm~\ref{alg:constructH_abelian-2BGA} provides candidate subgroups $\set{H}\leq\set{S}$ for Algorithm~\ref{alg:LocalASC}. Preservation of the vertex classes and Tanner graph incidences does not guarantee that distinct edges incident on the same vertex belong to distinct edge orbits. This condition is checked separately before solving, as required by the proper-edge-coloring condition in Definition~\ref{def:asc-constraints}. Only a lifted assignment that passes the final verification is accepted as an ASC-feasible schedule.

For the $[[36,4,4]]$, $[[36,4,6]]$, and $[[108,12,6]]$ codes in Table~\ref{tab:LocalASC-vs-ASC_depth6}, Algorithm~\ref{alg:constructH_abelian-2BGA} returns $\set{C}_3\times\set{C}_6$, $\set{C}_3\times\set{C}_6$, and $\set{C}_{54}$, respectively. The orbit-reduced models of these groups are infeasible at depth $6$, but those of their odd-order subgroups yield depth-$6$ schedules without computing $\Aut(\graph{G})$. For the $[[108,12,6]]$ code, this reduces the average LocalASC runtime from $769$ to $31$ milliseconds.

\subsection{Automorphism Groups of Quantum Tanner Codes}
\label{sec:qTan-aut}

Consider a classical Tanner code $\code C$ with parity-check matrix $\tn{Tan}(\graph{B}, \mat{H})$ with bits on the variable nodes of an $(l, s)$-biregular base graph $\graph{B} = (\set V_{\tn{c}}\sqcup\set{V}_{\tn{d}}, \set E)$ and constraints given by an $m\times l$ parity-check matrix $\mat H$ with Tanner graph $\graph{G}^\mat{H}=(\set{V}_{\tn{c}}^\mat{H}\sqcup \set{V}_{\tn{d}}^\mat{H}, \set{E}^\mat{H})$, at each local view. The code $\code C$ has Tanner graph $\graph{G}$ with vertex set $\set{V}_\tn{x} \sqcup \set{V}_\tn{d}$ where $\set{V}_\tn{d}$ are the variable nodes and $\set{V}_\tn{x} = \set{V}_{\tn{c}}\times\set{V}_{\tn{c}}^\mat{H}$ are the check nodes. 
An automorphism $f\in \Aut(\graph{B})$ induces a map 
\begin{IEEEeqnarray*}{c}
  \begin{tikzcd}[row sep=0mm]
    \tilde{f}: \set{V}_\tn{x} \sqcup \set{V}_\tn{d} \ar[r]
    &\set{V}_\tn{x} \sqcup \set{V}_\tn{d}
    \\
    \hspace{3mm}\set{V}_\tn{x} \ni (v,t) \ar[r, mapsto]
    &(f(v),t)
    \\
    \hspace{3mm}\set{V}_\tn{d} \ni q \ar[r, mapsto]
    &f(q) \rlap{.}
  \end{tikzcd}
\end{IEEEeqnarray*}
As a function on vertex sets, $\tilde{f}$ is well-defined and invertible. More generally, we can, given a permutation $\sigma_v: \set{V}_\tn{c}^\mat{H} \to \set{V}_\tn{c}^\mat{H}$ for each $v\in \set{V}_\tn{c}$, define the function
\begin{IEEEeqnarray*}{c}
  \begin{tikzcd}[row sep=0mm]
    \tilde{f}_\sigma\colon\set{V}_\tn{x} \sqcup \set{V}_\tn{d} \ar[r]
    &\set{V}_\tn{x} \sqcup \set{V}_\tn{d}
    \\
    \hspace{3mm}\set{V}_\tn{x} \ni (v, t) \ar[r, mapsto]
    &(f(v), {\sigma_v(t)})
    \\
    \hspace{3mm}\set{V}_\tn{d} \ni q \ar[r, mapsto]
    &f(q) \rlap{.}
  \end{tikzcd}
\end{IEEEeqnarray*}
Next, we will consider sufficient conditions for $\tilde{f}$ and $\tilde{f}_{\sigma}$ to be graph homomorphisms of $\graph{G}$, and hence Tanner graph automorphisms for $\code{C}$.



\begin{proposition}
  \label{prop:automorphisms-classical-Tan}
  Let $f\in \Aut(\graph{B})$ and permutations $\sigma_v$ on $\set{V}_\tn{c}^\mat{H}$ be given, and define $\tilde{f}_\sigma$ as above. Then 
  \begin{equation*}
    \tilde{f}_\sigma\in \Aut(\graph{G}) \iff \sigma_v\sqcup f\at{\delta(v)}\in \Aut(\graph{G}^\mat{H}) \tn{ for all } v\in \set{V}_\tn{c}.
  \end{equation*}
\end{proposition}
\begin{IEEEproof}
  See Appendix~\ref{sec:proof_automorphisms-classical-Tan}.
\end{IEEEproof}

The automorphisms $\tilde{f}_\sigma$ can be thought of as symmetries $f$ of the base graph $\graph{B}$ and symmetries $\psi_v$ of the Tanner graph $\graph{G}^\mat{H}$ that fit together, in the sense that $\psi_v = \sigma_v \sqcup f\at{\delta(v)}$. 
In principle, one can also have automorphisms that come from symmetries between $\graph{G}^\mat{H}$ and $\graph{B}$, but we will not consider this case further here, as it is very restrictive.

\begin{corollary}
  \label{cor:nice-automorphisms}
  $\tilde{f}\in \Aut(\graph{G})$ if and only if $\operatorname{id}\sqcup f\at{\delta(v)}\in \Aut(\graph{G}^\mat{H})$ for all $v\in\set{V}_\tn{c}$. In particular, $\tilde{f}\in \Aut(\graph{G})$ if $f\in \Aut(\graph{B})$ is label-preserving.
\end{corollary}
\begin{IEEEproof}
  See Appendix~\ref{sec:proof_nice-automorphisms}.
\end{IEEEproof}

The definition of $\tilde{f}_\sigma$ generalizes straightforwardly to quantum Tanner codes, where we have two local codes instead of one. We write $\graph{I}(\set{X}) = \graph{I}(\graph{G}_0^\square) \cup \graph{I}(\graph{G}_1^\square)$ and call this the \textit{vertex-square incidence graph} of $\set{X}$. In $\Aut(\graph{I}(\set{X}))$, we only consider automorphisms that map $\set{V}_0$ to $\set{V}_0$, $\set{V}_1$ to $\set{V}_1$, and $\set{Q}$ to $\set{Q}$.
\cref{thm:aut-QT} then follows the same way as \cref{prop:automorphisms-classical-Tan}.

\begin{theorem}
  \label{thm:aut-QT}
  Given $f\in \Aut(\graph{I}(\set{X}))$ and permutations $\sigma_v$ on the check nodes of $\graph{G}_{\mat{G}_\tn{A}\otimes \mat{G}_\tn{B}}$ for each vertex $v$ of $\set{V}_0$ and on the check nodes of $\graph{G}_{\mat{H}_\tn{A}\otimes \mat{H}_\tn{B}}$ for each $v\in \set{V}_1$, the map $\tilde{f}_\sigma$ is an automorphism of the quantum Tanner code defined on $\set{X}$ using $\mat{H}_\tn{A}\otimes \mat{H}_\tn{B}$ and $\mat{G}_\tn{A}\otimes \mat{G}_\tn{B}$ as local restrictions if and only if
  \begin{IEEEeqnarray*}{rCl}
    \sigma_v\sqcup f\at{\delta(v)}& \in &\Aut(\graph{G}_{\mat{G}_\tn{A}\otimes \mat{G}_\tn{B}})\tn{ for all } v\in\set{V}_0, \,\,\tn{ and}
    \\
    \sigma_v\sqcup f\at{\delta(v)}& \in &\Aut(\graph{G}_{\mat{H}_\tn{A}\otimes \mat{H}_\tn{B}})\tn{ for all } v\in\set{V}_1,
  \end{IEEEeqnarray*}
  where $\delta(v) \subseteq \set{E}(\graph{I}(\set{X}))$.
\end{theorem}

It is sometimes easier to picture the graphs $\graph{G}^\square_0$, $\graph{G}^\square_1$ rather than the square complex $\mathcal{X}$, and automorphisms on $\graph{I}(\graph{G}^\square_0)$ and $\graph{I}(\graph{G}^\square_1)$ form automorphisms on $\graph{I}(\set{X})$ if they agree on $\set{Q}$. Quantum Tanner codes from the bipartite construction have $\graph{G}^\square_0 \cong \graph{G}^\square_1$, with the isomorphism given, for example, by the identity map on the underlying group when $\set{X}$ is a left-right Cayley complex. However, the map $\phi$ from the edges of $\graph{G}^\square_0$ to the edges of $\graph{G}^\square_1$ induced by the construction is not a graph automorphism. The connection between the automorphism group of the square complex and these graphs is given in \cref{prop:Aut-X-to-Gs}. 
In particular, an automorphism $g\in \Aut(\graph{I}(\set{X}))$ is uniquely defined by its restriction to $\graph{I}(\graph{G}_0^\square)$.

\begin{proposition}
  \label{prop:Aut-X-to-Gs}
  An automorphism $f\in\Aut(\graph{I}(\graph{G}^\square_0))$ is a restriction of an automorphism on $\graph{I}(\mathcal{X})$ if and only if $\phi f|_\set{Q} \inv{\phi}$ induces a function on the vertices of $\graph{G}^\square_1$. 
  In that case, the induced function is the only extension of $f$ to $\graph{I}(\mathcal{X})$.
\end{proposition}
\begin{IEEEproof}
  $f$ is an invertible function on the vertices of $\set{V}_0$ and the squares of $\set X$. To extend $f$ to $\set{X}$, squares sharing vertices in $\set{V}_1$ must be mapped to squares sharing vertices in $\set{V}_1$. The induced map is $\phi f\inv{\phi}$ on the edges of $\graph{G}^\square_1$.
\end{IEEEproof}



When $\set X$ is a left-right Cayley complex of an abelian group $\mathcal{G}$, the group $\mathcal{G}$ acts on $\set{X}$ by group multiplication in a way that preserves edge labels. For example, an edge $(g, ag)$ is mapped by $h\in \mathcal{G}$ to an edge $(hg, ahg)$, both of which are labeled by $a$ (and $a^{-1}$). The induced map $\tilde{h}$ on $\graph{G}$ therefore respects any local codes, leading to \cref{thm:autQT-ab-Cay}. Here, we say that a quantum Tanner code is constructed from the group $\mathcal{G}$ if the underlying square complex $\set{X}$ is a left-right Cayley complex on $\mathcal{G}$ such as in \cref{sec:prelim-QT}. In \cite{MostadRosnesLin25_1}, it is shown that, in a more general setting, there is still an underlying group, in which case we say the code is constructed from that group. Also, recall that quadripartite codes constructed from the group $\set{G}$ can be seen as codes from the bipartite construction on $\set{G}\times \set{C}_2$.

\begin{theorem}
  \label{thm:autQT-ab-Cay}
  Let $\graph{G}$ be the Tanner graph of a quantum Tanner code constructed from the group $\mathcal{G}$ using the bipartite construction. If $\mathcal{G}$ is abelian, then $\mathcal{G}\subseteq \Aut(\graph{G})$.
\end{theorem}

More generally, let $C_{\set G}(\set A) \subseteq \set G$ denote the centralizer of $\set A$ in $\set G$, that is, $C_{\set G}(\set A) = \{g \in \set G : ag = ga \tn{ for all } a \in \set A\}$. Then $C_{\set G}(\set A)$ acts on the square complex $\set X$ on the left, and $C_{\set G}(\set B)$ acts on $\set X$ on the right. 

\begin{example}
  \label{ex:QT-45-7-4}
  A $[[45,7,4]]$ quantum Tanner code can be constructed as follows. Take $\mathcal{G}=\set{C}_{10}$, $\set{A}=\{7,0,3\}$, $\set{B} = \{4,5,6\}$, and 
  \begin{equation*}
    \mat{H}_\tn{A} = \mat{G}_\tn{B} = \begin{bmatrix}
      1&0&1 \\ 0&1&1
    \end{bmatrix}, 
    \quad
    \mat{G}_\tn{A} = \mat{H}_\tn{B} = \begin{bmatrix}
      1&1&1
    \end{bmatrix}.
  \end{equation*}
  This leads to a maximum row and column weight of $\Delta=6$, as we have made sure that the third element of $\set{A}$ and $\set{B}$, corresponding to the heavy column of $\mat{H}_\tn{A}$ and $\mat{G}_\tn{B}$, is not self-inverse. Using the (original) bipartite construction, we get $n=\frac{1}{2}|\set{G}||\set{A}||\set{B}| = 45$ and $k\geq \frac{1}{2}|\set{G}|\left(|\set{A}||\set{B}| - 4 k_{\textnormal{A}}k_\textnormal{B}\right)=5$. Since $\mathcal{G}$ is abelian, \cref{thm:autQT-ab-Cay} tells us that $\mathcal{G}\subseteq \Aut(\graph{G})$. One can construct a depth-$6$ schedule for the code using LocalASC with a subgroup $\set{C}_5\subseteq\set{G}$, thereby reducing the coloring problem from $240$ to $48$ variables.

  The $[[325,14,9]]$, $[[666,4,d\geq 13]], [[666,8,d\geq 13]]$, $[[1225,32,d\geq 13]]$, and $[[1225,225,6]]$ quantum Tanner codes are constructed similarly to the $[[45,7,4]]$ code, in such a way that $\graph{G}_0^\square$ and $\graph{G}_1^\square$ both are complete graphs. In particular, these graphs have no parallel edges, which is believed to be beneficial for belief propagation decoders~\cite{MostadRosnesLin26_1sub}. The underlying groups are all cyclic, meaning that \cref{thm:autQT-ab-Cay} applies, giving the automorphism groups $\set{C}_{26}$, $\set{C}_{37}$, $\set{C}_{37}$, $\set{C}_{50}$, and $\set{C}_{50}$, respectively. Note that none of these codes are quadripartite, so the approach of~\cite{NguyenRimbach-RussBosco26_1sub} does not apply to them.
\end{example}

In \cref{ex:QT-45-7-4}, the quantum Tanner codes are constructed from Cayley graphs following the construction described in \cref{sec:prelim-QT}. Similar to how coset codes generalize two-block group-algebra codes, we can construct quantum Tanner codes from certain non-Cayley graphs. Then the subgroup obtained from \cref{thm:autQT-ab-Cay} will be somewhat smaller.


\begin{example}\label{ex:QT-180-26-6}
  Consider the $[[180, 26, 6]]$ code from \cite{MostadRosnesLin25_1}. It is constructed using the same local codes $\mat{H}_\tn{A}, \mat{H}_\tn{B}, \mat{G}_\tn{A}, \mat{G}_\tn{B}$ as in \cref{ex:QT-45-7-4}, leading to a maximum row and column weight of $\Delta(\graph{G})=6$. 
  The square complex $\mathcal{X}$ is constructed using the Petersen graph, a non-Cayley graph made by joining two $5$-cycles to make a $3$-regular graph with $10$ vertices and girth $5$. As such, \cref{thm:autQT-ab-Cay} tells us that $\set{C}_{5}\times \set{C}_2\cong\set{C}_{10}\subseteq \Aut(\graph{G})$, since the quadripartite construction is used. 
  LocalASC finds a depth-$6$ schedule for syndrome extraction using both the subgroups $\set{C}_5$ and $\set{C}_{10}$.
  
\end{example}

\begin{remark}
  A partition of the $\set V_\tn{X}$, $\set V_\tn{D}$, and $\set V_\tn{Z}$ vertices of a Tanner graph $\graph{G}$ defines a quotient graph $\quotient{\graph{G}}{\mathcal{R}}$ that can be viewed as a Tanner graph of a CSS code. As all edges of $\graph{G}$ go between data vertices and check vertices, each edge of $\graph{G}$ is mapped to an edge of the quotient graph. Therefore, if no two edges sharing an endpoint are mapped to the same edge in the quotient graph, an edge coloring on $\quotient{\graph{G}}{\mathcal{R}}$ defines an edge coloring on $\graph{G}$. However, a quantum-valid ordered coloring on $\quotient{\graph{G}}{\mathcal{R}}$ does not necessarily give a quantum-valid ordered coloring on $\graph{G}$ (see \cref{sec:NOT-enough_proper-edge-coloring} for an example). We therefore instead consider the graph $\graph{G}$ directly, demanding that edges identified in $\quotient{\graph{G}}{\mathcal{R}}$ have the same color. In principle, any partition on the edges of $\graph{G}$ respecting $\tn{X}$ and $\tn{Z}$ edges could be used. When the partition comes from a partition of vertices, the approach is equivalent to adding certain extra constraints on the ordered coloring of $\graph{G}/\mathcal{R}$ to make it lift to a quantum-valid coloring of $\graph{G}$. Each even parity constraint on $\graph{G}$ maps to an even parity constraint on $\graph{G}/\mathcal{R}$, possibly with many redundant constraints.
\end{remark}

For some quantum Tanner codes, there are apparent ways to partition the edges of their Tanner graphs. For quantum Tanner codes defined on a quadripartite square complex 
$$\set X = (\set{V}_{00}\sqcup\set{V}_{10}\sqcup\set{V}_{01}\sqcup\set{V}_{11}, \set{E}_\tn{A}\sqcup\set{E}_\tn{B}, Q),$$ 
such as those constructed using the quadripartite construction, the following partition does not identify adjacent edges, and is therefore valid. Check nodes of $\graph{G}$ coming from the same set of vertices $\set{V}_{ij}$ and the same row of the corresponding local parity-check matrix are identified, as well as variable nodes with the same labels on the same vertex sets, i.e., $((g,00), (ag,10), (gb,01), (agb,11)) \sim ((h,00), (ah,10), (hb,01), (ahb,11))$.

The edges of quantum Tanner codes where $\set A$ and $\set B$ contain no self-inverse elements can be partitioned according to whether they connect to an $\tn{X}$ or $\tn{Z}$ type check and their label, which is the label of the variable node in the local view of the check node it connects to. Note that this partition is not necessarily induced by any partition of the vertices.


In the last case, a quantum-valid coloring of the two 0-rate CSS codes $(\mat{H}_\tn{A}, \mat{G}_\tn{A})$, $(\mat{H}_\tn{B}, \mat{G}_\tn{B})$ will often lift to a quantum-valid coloring of the quantum Tanner code. The depth of this circuit is the product of the two circuit depths, and since the column and row weights are products of the local codes' column and row weights, the resulting circuit may be depth-optimal. In this way, one can design codes with optimal-depth circuits by choosing local codes with optimal-depth circuits. 

\begin{theorem}
  \label{thm:achievable-depth_qTanner-local-codes}
  Quantum-valid ordered colorings for the CSS codes $(\mat{G}_\tn{A}, \mat{H}_\tn{A})$, $(\mat{G}_\tn{B}, \mat{H}_\tn{B})$ of depth $\lambda_1$ and $\lambda_2$ define a quantum-valid coloring of depth $\lambda_1 \lambda_2$ for any quantum Tanner code with local codes defined by
  \begin{equation*}
    \code C_\tn{A} = \tn{ker}(\mat{H}_\tn{A}) = \tn{im}(\trans{\mat{G}}_\tn{A}),
    \quad
    \code C_\tn{B} = \tn{ker}(\mat{H}_\tn{B}) = \tn{im}(\trans{\mat{G}}_\tn{B}),
  \end{equation*}
  if the local view of variable node $b$ in $\graph{G}(\mat{G}_\tn{B}, \mat{H}_\tn{B})$ shares no colors with the local view of $\inv{b}$ for all $b\in \set{B}$.
\end{theorem}
\begin{IEEEproof}
  See Appendix~\ref{sec:proof_achievable-depth_qTanner-local-codes}.
\end{IEEEproof}

For there to exist a coloring of $\graph{G}(\mat{G}_\tn{B}, \mat{H}_\tn{B})$ where $b$ and $\inv{b}$ do not share any colors, we need $\set{B}$ to have no self-inverse elements. This is the case, for example, for the $[[666, 8, d\geq 13]]$ quantum Tanner code. For $\set{B}$ to contain no self-inverse elements, it must have even order, which is the case for example for all codes in~\cite{LeverrierRozendaalZemor25_1sub}, which are also quadripartite. We suspect it may be possible to apply a similar strategy to that in~\cref{thm:achievable-depth_qTanner-local-codes} to the quadripartite case.

\section{Numerical Results}
\label{sec:numerical-results}

We apply LocalASC and ASC to selected pairs of measured CSS PCMs. Both LocalASC and ASC use the CP-SAT solver in OR-Tools~\cite{PerronDidier26_1soft}. The LocalASC runtimes reported in Table~\ref{tab:LocalASC-vs-ASC_depth6} include computing $\Aut(\graph{G})$ using \texttt{nauty}~\cite{McKayPiperno14_1}, selecting a subgroup, and solving the resulting orbit-reduced model. The reported ASC runtime includes building and solving the unreduced model with the solver of~\cite{Zhang-etal26_1sub}. Because the parallel CP-SAT search is nondeterministic, the reported LocalASC and ASC runtimes for each completed instance are averaged over $10$ independent runs. Instances for which no ASC-feasible schedule is found within the $30$ minutes time limit are not repeated and are reported as $>30$\,min. All computations were performed on a single machine with an $8$-core Apple M$3$ processor and $24$~GB of memory. Moreover, we verified all exact minimum distances reported in this paper using the distance-computation methods of~\cite{ChenJafariLai26_1sub}. Table~\ref{tab:LocalASC-vs-ASC_depth6} lists the CSS code instances for which LocalASC finds an orbit assignment whose lift satisfies all QCEC constraints. Consequently,
$\Qchi(\mat{H}_{\tn{X}},\mat{H}_{\tn{Z}})=w=\Delta(\graph{G})$
for every entry. Each schedule is lower-bound-saturating and depth-optimal, while each pair of measured CSS PCMs is weight-optimal and has $\Qgap=0$. The code and the PCMs in Table~\ref{tab:LocalASC-vs-ASC_depth6} are available at~\cite{LocalASC26_1}.
 
To assess the circuit-level performance of the syndrome-extraction schedules, we simulate memory circuits using \texttt{stim}~\cite{Gidney21_1}. For $\tn{Z}$ (resp. $\tn{X}$)-basis experiments, we initialize the data qubits in the $\ket{0}$ (resp. $\ket{+}$) state and perform $d$ rounds of syndrome extraction following the obtained schedule for a quantum code of minimum distance $d$. The $\tn{Z}$-check (resp. $\tn{X}$-check) ancillas are initialized in $\eket{0}$ (resp. $\eket{+}$) and act as the target (resp. control) of the CNOT gates. Ancillas are measured and reinitialized at each round. We perform a final readout of all data qubits in the corresponding basis at the end of the circuit. For a fair comparison with the numerical implementation of coset-based codes~\cite{AydinTamoBarg26_1sub}, we also do not consider a perfect final round in the experiments, in contrast to other implementations of memory experiments~\cite{Bravyi-etal24_1}.
  
Since we extract both $\tn{X}$- and $\tn{Z}$-check syndromes, the resulting detector error model~\cite{DerksTownsend-TeagueBurchardsEisert25_1} becomes considerably large, particularly under the circuit-level noise model we consider. Here, all single- and two-qubit gates, state initializations, measurements, and idle qubits in any circuit tick are prone to errors. All noise locations have the same physical error rate $p$. Specifically, single-qubit gates and idle qubit locations are followed by a single-qubit depolarizing channel with probability $p$; two-qubit gates are followed by a two-qubit depolarizing channel with probability $p$; initializations in the $\tn{Z}$ (resp. $\tn{X}$) basis are followed by an $\tn{X}$ (resp. $\tn{Z}$) error with probability $p$, and measurement outcomes are flipped with probability $p$. Note that for some codes, there are ways to reduce idle positions over several rounds of syndrome extraction by delaying the initialization of one of the groups of ancillas ($\tn{Z}$ or $\tn{X}$) and starting to perform CNOTs while the second group of ancillas is initialized, as done in the syndrome-extraction circuits of BB codes~\cite{Bravyi-etal24_1} or in the construction of left-right circuits~\cite{StrikisBrowneBeverland26_1sub}.
  
For performance benchmarking, we run memory experiments as described above using Higgott's \texttt{stimbposd}~\cite{Higgott23_1} implementation of belief propagation with ordered statistics decoding and combination-sweep of order $10$ (BP+OSD-CS10). We use a maximum of $1000$ BP iterations. Each $\tn{X}$- or $\tn{Z}$-memory experiment consists of two independent runs, each stopped after $50$ logical errors or $10^{6}$ samples. Within each comparison, the experiments use the same code and differ only in the syndrome-extraction schedule. The logical error rate (LER) per cycle in a $\tn{W}$-basis experiment we compute is defined by Aydin \textit{et al.}~\cite[Sec.~VI]{AydinTamoBarg26_1sub} as
\begin{equation*}
  p_{\tn{L},\tn{W}}(p) = 1 -\bigg(1 - \frac{N_{\tn{e},d,\tn{W}} (p)}{N_{\tn{s},d,\tn{W}}(p)} \bigg)^{1/d},\quad \tn{W} \in \{\tn{X},\tn{Z}\}.
\end{equation*}
Here, $N_{\tn{e},d,\tn{W}}(p)$ is the number of logical errors that occurred in a Monte Carlo simulation of the $\tn{W}$-basis experiment with $N_{\tn{s},d,\tn{W}}(p)$ simulated trials and physical error rate $p$. We also report the LER per cycle as the complement to the simultaneous success probability
\begin{equation*}
  p_\tn{L} = 1 - (1-p_{\tn{L},\tn{X}})(1-p_{\tn{L}, \tn{Z}}).
\end{equation*}
As in~\cite{AydinTamoBarg26_1sub}, the computed pseudo-threshold solves the break-even condition
\begin{IEEEeqnarray*}{c}
  p_\tn{L} (p_{\tn{th}}) = 1 -(1-p_{\tn{th}})^{k}
\end{IEEEeqnarray*}
and the heuristic fitting formula for the sub-threshold regime is quadratic in $p$ in the exponent
\begin{equation*}
  \tilde{p}_\tn{L}(p) = p^{d_{\tn{circ}} /2} \ope^{\alpha + \beta p + \gamma p^2},
\end{equation*}
where $d_\tn{circ}$ is the \textit{overall} circuit distance defined below. Since the schedules produced by LocalASC are solver dependent, we go through the same circuit filtering step as Zhang \textit{et al.}~\cite{Zhang-etal26_1sub}, but without a fixed random seed for the noise sampling. In this step, we run LocalASC a number of times. For each schedule found, we run a decoding experiment with $10^{5}$ trials at a single noise level $p=0.001$ and a maximum of $10$ BP iterations. At the end, we select the schedule that induced the lowest LER per cycle. For Fig.~\ref{fig: 48_4_8_filtering}, we ran the filtering for $50$ steps to select the best performing schedule. The results of the memory experiments for the $[[48,4,8]]$ coset-based code~\cite{AydinTamoBarg26_1sub} reported in Fig.~\ref{fig:48_4_d8_depth-schedule-comparison} use the schedules obtained from the filtering step. We consider only points with at least $5$ logical errors, summed over the $\tn{X}$- and $\tn{Z}$-memory experiments. The extrapolation uses only those below the computed pseudo-threshold. 
  
To analyze how different schedules preserve the code's error-correcting capabilities, we consider the circuit distance, as in~\cite{StrikisBrowneBeverland26_1sub}. In general, the circuit distance is the minimum weight of an error $\vect{e}$ that does not flip any detectors but flips at least one observable~\cite{DerksTownsend-TeagueBurchardsEisert25_1}. That is, for a detector error matrix $\mat{H}_{\tn{dem}}$ and a corresponding observables matrix $\mat{L}$, the circuit distance is given by $\min |\vect{e}|$ such that $\vect{e} \in \ker(\mat{H}_{\tn{dem}})\setminus \ker(\mat{L})$. Strikis \textit{et al.} build on this definition and use the \textit{overall} circuit distance $d_{\tn{circ}} = \min\{d_{\tn{X},\tn{circ}}, d_{\tn{Z},\tn{circ}}\}$, as it determines the leading-order scaling of the logical failure probability. Here, $d_{\tn{W}, \tn{circ}}$ is the circuit distance of a $\tn{W}$-basis memory experiment with $d$ rounds of syndrome extraction for $\tn{W}\in\{\tn{X},\tn{Z}\}$. In the simulations, we estimate the circuit distance by using \texttt{QDistRndMW} from Webster's \texttt{codeDistance} package~\cite{WebsterJacobHiggott26_1sub} running at most $5000$ iterations. Using this, we found an upper bound of $5$ for all $[[48,4,8]]$ schedules, which we use in the extrapolation. 
  
  
\begin{figure}[htbp!]
  \centering
  \begin{tikzpicture}
    \begin{semilogyaxis}[
      width=8.5cm,
      height=6cm,
      ymin=4e-5, ymax=8e-5,
      log origin=infty,
      xtick={0,10,20,30,40,49},
      xtick distance=10,
      xtick={0,10,20,30,40,49},
      ytick={4e-5,5e-5,6e-5,7e-5,8e-5},
      yticklabels={$4$,$5$,$6$,$7$,$8$},
      minor x tick num=1,
      ylabel={LER per cycle ($\times 10^{-5}$)},
      xlabel={Schedules ordered by estimated LER},
      grid=major,
      grid style={dashed, draw=gray!30},
      enlarge x limits=0.02
      ]
      
      \addplot [
      ybar,                 
      bar width=4.5pt,       
      fill=blue!20,     
      draw=blue!40!gray,   
      error bars/.cd,
      y dir=both, 
      y explicit,
      error bar style={thick, solid, black!80} 
      ] table [
      x=x_index, 
      y=total_pL, 
      y error=total_pL_err, 
      col sep=comma
      ] {plot_data/filtered_48_4_8_circuits.csv};
      
      \addplot [
      color=cyan!60!green, 
      thick,
      mark=*,
      mark size=1pt,
      mark options={fill=cyan!60!green, draw=cyan!60!green}
      ] table [
      x=x_index, 
      y=total_pL, 
      col sep=comma
      ] {plot_data/filtered_48_4_8_circuits.csv};
    \end{semilogyaxis}
  \end{tikzpicture}
  \caption{LocalASC filtering with $50$ steps for the $[[48,4,8]]$ coset-based code at physical noise level $p=0.001$.}
  \label{fig: 48_4_8_filtering}
\end{figure}
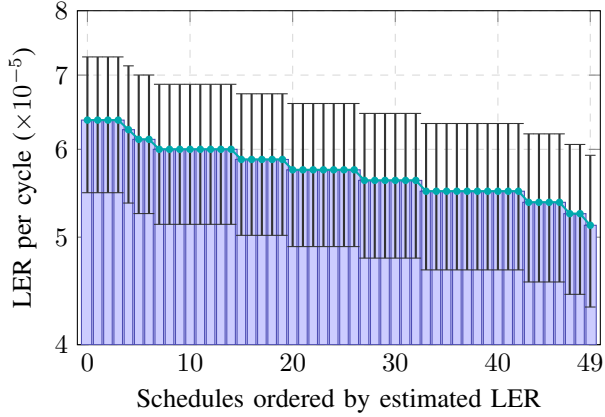

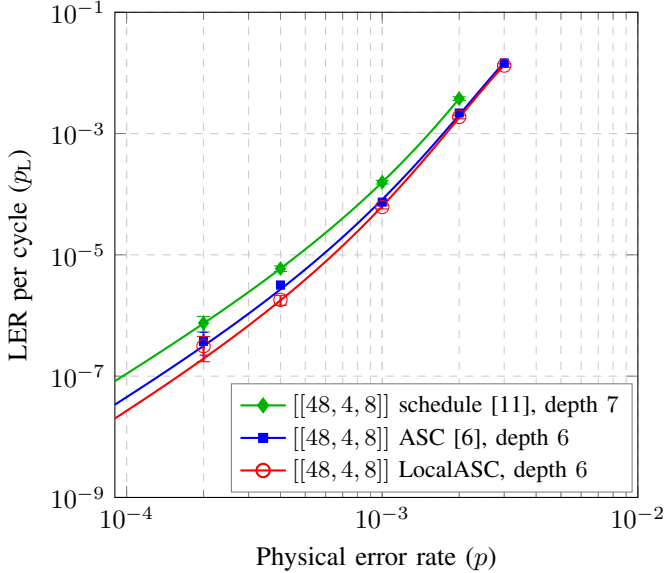
\begin{figure}[htbp!]
  \centering
  \begin{tikzpicture}
    \begin{loglogaxis}[
      width=8.5cm,
      height=8cm,
      xmin=9e-5, xmax=0.01,
      ymin=1e-9, ymax=1e-1,
      xlabel={Physical error rate ($p$)},
      ylabel={LER per cycle ($p_\tn{L}$)},
      grid=both,
      grid style={dashed, draw=gray!40},      
      xtick={1e-4, 1e-3, 1e-2},
      ytick={1e-13,1e-11, 1e-9, 1e-7, 1e-5, 1e-3, 1e-1},
      legend pos=south east,
      legend style={
        at={(0.99,0.01)},
        anchor=south east,
        font=\small,
        cells={anchor=west},
        draw=gray},      
      ]
      
      
      \addplot [
      color=green!70!black, 
      thick, mark=none,
      legend image post style={mark=diamond*, mark size=2.5pt}
      ] table [x=p, y=pL, col sep=comma] {plot_data/48_4_d8_ATB_fit.csv};
      \addlegendentry{$[[48, 4, 8]]$ schedule~\cite{AydinTamoBarg26_1sub}, depth $7$}
      
      \addplot [
      color=green!70!black,
      only marks, forget plot,
      mark=diamond*, mark size=2.5pt,
      error bars/.cd,
      y dir=both, y explicit,
      error bar style={thick, solid}
      ] table [
      x=p, 
      y=total_pL, 
      y error=total_pL_err, 
      col sep=comma
      ]{plot_data/48_4_d8_ATB_data.csv};
      
      
      \addplot [
      color=blue, 
      thick, mark=none,
      legend image post style={mark=square*, mark size=1.5pt},      
      ] table [x=p, y=pL, col sep=comma] {plot_data/48_4_d8_fullASC_fit.csv};
      \addlegendentry{$[[48, 4, 8]]$ ASC~\cite{Zhang-etal26_1sub}, depth $6$}
      
      \addplot [
      color=blue,
      only marks, forget plot,
      mark=square*, mark size=1.5pt,
      error bars/.cd,
      y dir=both, y explicit,
      error bar style={dashed,thick}
      ] table [
      x=p, 
      y=total_pL, 
      y error=total_pL_err, 
      col sep=comma
      ]{plot_data/48_4_d8_fullASC_data.csv};
      
      
      \addplot [
      color=red, 
      thick, mark=none,
      legend image post style={mark=o, mark size=2.5pt}
      ] table [x=p, y=pL, col sep=comma] {plot_data/48_4_d8_LocalASC_fit.csv};
      \addlegendentry{ $[[48, 4, 8]]$ LocalASC, depth $6$}
      
      \addplot [
      color=red,
      only marks, forget plot,
      mark=o, mark size=2.5pt,
      error bars/.cd,
      y dir=both, y explicit,
      error bar style={thick, solid}
      ] table [
      x=p, 
      y=total_pL, 
      y error=total_pL_err, 
      col sep=comma
      ]{plot_data/48_4_d8_LocalASC_data.csv};              
      
    \end{loglogaxis}
  \end{tikzpicture}
  \vspace{-3ex}
  \caption{The same coset-based $[[48,4,8]]$ code memory experiments using BP+OSD-CS10 for circuits obtained from different scheduling methods. Circuits were selected after a filtering step when possible.} 
  \label{fig:48_4_d8_depth-schedule-comparison}
\end{figure}

In Fig.~\ref{fig:48_4_d8_depth-schedule-comparison}, the depth-$6$ LocalASC schedule selected by circuit filtering gives a lower estimated LER than both the depth-$6$ ASC schedule and the depth-$7$ schedule of Aydin \emph{et al.}~\cite{AydinTamoBarg26_1sub} at the same simulated noise levels.

For the $[[72,12,6]]$ BB code, we compare the depth-$7$ syndrome-extraction schedule of Bravyi \emph{et al.}~\cite{Bravyi-etal24_1} with a depth-$7$ LocalASC schedule using the same noise model and decoder settings. The LocalASC schedule gives an estimated LER comparable to that of the schedule of Bravyi \emph{et al.} in these simulations. This benchmark complements the $[[48,4,8]]$ comparison by considering an instance with $\Qchi=7>\Delta=6$. The results are presented in Appendix~\ref{sec:LocalASC_versus_IBM_BB_72_12_d6}.

\section{Discussion and Conclusion}
\label{sec:discussion-conclusion}

In this paper, we formulated low-depth CSS syndrome extraction as an ordered edge-coloring problem with quantum parity constraints. The resulting QCEC number equals the minimum two-qubit depth in the standard single-ancilla model. We then introduced LocalASC, which reduces the QCEC constraint system to edge orbits under a vertex-class-preserving automorphism subgroup and lifts a feasible orbit assignment to the full Tanner graph. LocalASC finds depth-$w$ schedules for $18$ pairs of measured CSS PCMs from two-block and quantum Tanner codes, with $\Qchi=w=\Delta(\graph{G})$ and $w\in\{6,8,9,12\}$. All $18$ PCM pairs are weight-optimal and have zero quantum scheduling gap.

Moreover, QCEC and LocalASC depend on the measured CSS PCMs rather than solely on the stabilizer group they generate. Removing dependent checks preserves the row spaces but may change the Tanner graph and its automorphism group. Conversely, dependent checks may be added to increase the available symmetry, provided that the original checks are retained and the maximum Tanner graph degree $\Delta(\graph{G})$ is unchanged. Any schedule for the enlarged pair of PCMs can then be restricted to the original checks without increasing its depth. This observation motivates the joint selection of bounded-weight stabilizer generators and a quantum-valid edge coloring for them.

We remark that the infeasibility of the orbit-reduced ASC model at a given depth does not imply the infeasibility of the unreduced ASC model at the same depth. For the $7$ published IBM BB codes~\cite{Bravyi-etal24_1}, the translation subgroup reduces each scheduling instance to $12$ edge-orbit variables, independent of the code length, and LocalASC rapidly finds and verifies a depth-$7$ schedule. The optimality of these schedules is established by combining the feasible depth-$7$ schedules obtained using LocalASC with solver-certified infeasibility of the corresponding unreduced ASC instances at depth $6$, obtained using the solver of~\cite{Zhang-etal26_1sub}.\footnote{The abstract of~\cite{Zhang-etal26_1sub} states that the compiler ``certifies that no depth-$6$ syndrome-extraction circuit exists'' for these codes, whereas Sec.~V-A of the same paper describes the result more cautiously as ``concrete evidence supporting the conjecture.'' In the released implementation, \texttt{find\_min\_depth()} does not distinguish solver-certified infeasibility from a timeout: \texttt{solve()} returns \texttt{None} for both the \texttt{INFEASIBLE} and \texttt{UNKNOWN} statuses, after which the search proceeds to the next depth. In the computations, we record the CP-SAT status explicitly and accept only \texttt{INFEASIBLE} as establishing infeasibility at the corresponding depth.} 

Several research directions follow from this formulation. First, the code instances with $\Qchi=w=\Delta$ found here motivate constructing qLDPC codes whose measured CSS PCMs admit depth-$w$ syndrome extraction. We have begun to study such constructions and obtained preliminary results. Deriving algebraic conditions for $\Qgap=0$, especially for two-block coset-based and quantum Tanner code families, could help identify further weight-optimal instances when $w=\Delta$. Second, one may define a code-level invariant by minimizing $\Qchi$ over all pairs of CSS PCMs that generate the same stabilizer group and satisfy a prescribed check-weight bound. Finally, following Zhang \emph{et al.}~\cite{Zhang-etal26_1sub}, we use circuit filtering to select the schedule with the lowest estimated LER among those found at a fixed depth. Although the LERs of these schedules agree within the statistical error here, depth alone need not determine LER performance.
This motivates extending the ASC constraints either to limit hook error propagation or to impose a lower bound on circuit distance when constructing low-depth syndrome-extraction schedules.

\section*{Author Contributions}
E.~Durazo Rocha and O.~{\AA}.~Mostad contributed equally to this work. All authors contributed to the research and reviewed the manuscript. Generative AI tools were used to improve the clarity and presentation of the manuscript and to assist with intermediate mathematical checks and with the implementation, debugging, and checking of the LocalASC algorithms and numerical results. The authors independently verified all mathematical results, code, schedules, and numerical results through analytical checks, benchmarks, alternative implementations, and numerical computations. All authors take full responsibility for the work.

\appendices

\begin{figure*}[t!]
  \centering
  \subfloat[{The CSS Tanner graph of the $[[8,4,2]]$ code. The boxed qubit pairs are the common supports $\set{S}_{00}$, $\set{S}_{01}$, $\set{S}_{10}$, and $\set{S}_{11}$. $\set{S}_{00}$ is highlighted in orange.}]{%
    \begin{minipage}[t]{0.49\textwidth}
      \centering
      \begin{tikzpicture}[x=1cm,y=1cm,font=\fontsize{8}{9}\selectfont,
        line width=0.45pt]
        \node at (-3.2,3.12) {$\tn{X}$ checks};
        \node at (0,3.12) {data qubits};
        \node at (3.2,3.12) {$\tn{Z}$ checks};
        \coordinate (X0) at (-3.2,1.15);
        \coordinate (X1) at (-3.2,-1.15);
        \coordinate (Z0) at (3.2,1.15);
        \coordinate (Z1) at (3.2,-1.15);
        \foreach \i in {0,...,7}
        \coordinate (q\i) at (0,{2.25-1.35*floor(\i/2)-0.46*mod(\i,2)});
        \foreach \j/\support in {0/00,1/01,2/10,3/11} {
          \ifnum\j=0
            \def\supportcolor{orange!85!black}
          \else
            \def\supportcolor{gray!55}
          \fi
          \draw[draw=\supportcolor,rounded corners=2pt]
          (-0.29,{2.51-1.35*\j}) rectangle (0.29,{1.53-1.35*\j});
          \ifnum\j=0
            \node[text=orange!85!black,inner sep=0pt] at (0,{2.70-1.35*\j}) {$\set{S}_{\support}$};
          \else
            \node[inner sep=0pt] at (0,{2.70-1.35*\j}) {$\set{S}_{\support}$};
          \fi
        }
        \foreach \q in {q0,q1,q2,q3}
        \draw (X0)--(\q);
        \foreach \q in {q4,q5,q6,q7}
        \draw (X1)--(\q);
        \foreach \q in {q0,q1,q4,q5}
        \draw[dash pattern=on 2pt off 1.5pt] (Z0)--(\q);
        \foreach \q in {q2,q3,q6,q7}
        \draw[dash pattern=on 2pt off 1.5pt] (Z1)--(\q);
        \foreach \i in {0,1} {
          \node[rectangle,draw=black,fill=blue!12,
          minimum width=5.5mm,minimum height=4.5mm,inner sep=0pt]
          at (X\i) {$x_{\i}$};
          \node[rectangle,draw=black,fill=red!10,
          minimum width=5.5mm,minimum height=4.5mm,inner sep=0pt]
          at (Z\i) {$z_{\i}$};
        }
        \foreach \i in {0,...,7}
        \node[circle,draw=black,fill=gray!10,minimum size=4.1mm,inner sep=0pt]
        at (q\i) {$q_{\i}$};
      \end{tikzpicture}
    \end{minipage}%
    \label{fig:CSS-Tanner-graph_8_4_d2}%
  }\hfill
  \subfloat[{The quotient graph $\graph{G}/\set{H}_2$ for the action in~\eqref{eq:8_4_d2-good-rotation}. Each edge represents a size-$2$ edge orbit and carries the indicated orbit tick variable.}]{%
    \begin{minipage}[t]{0.49\textwidth}
      \centering
      \begin{tikzpicture}[x=0.92cm,y=1cm,font=\fontsize{8}{9}\selectfont,line width=0.45pt,
        qvertex/.style={circle,draw=black,fill=gray!10,minimum size=4.8mm,inner sep=0pt},
        ticklabel/.style={text=black,inner sep=0pt}]
        \node at (-3.2,2.35) {$\tn{X}$-check orbit};
        \node at (0,2.35) {data-qubit orbits};
        \node at (3.2,2.35) {$\tn{Z}$-check orbit};
        \node[rectangle,draw=black,fill=blue!12,minimum width=5.5mm,minimum height=4.5mm,inner sep=0pt] (X) at (-3.2,0) {$\set{X}$};
        \node[rectangle,draw=black,fill=red!10,minimum width=5.5mm,minimum height=4.5mm,inner sep=0pt] (Z) at (3.2,0) {$\set{Z}$};
        \node[below=3mm,inner sep=0pt] at (X.south) {$\{x_0,x_1\}$};
        \node[below=3mm,inner sep=0pt] at (Z.south) {$\{z_0,z_1\}$};
        \foreach \k/\y/\a/\b in {0/1.65/0/6,1/0.55/1/7,2/-0.55/2/4,3/-1.65/3/5} {
          \node[qvertex] (Q\k) at (0,\y) {$\set{Q}_{\k}$};
          \node[below=1.2mm,inner sep=0pt] at (Q\k.south) {$\{q_{\a},q_{\b}\}$};
          \draw (X) -- node[pos=0.55,above=2pt,ticklabel] {$\alpha_{\k}$} (Q\k);
          \draw[dash pattern=on 2pt off 1.5pt] (Q\k) -- node[pos=0.45,above=2pt,ticklabel] {$\beta_{\k}$} (Z);
        }
      \end{tikzpicture}
    \end{minipage}%
    \label{fig:quotient-TG_8_4_d2}%
  }
  \caption{The CSS Tanner graph of the $[[8,4,2]]$ code and its orbit reduction under $\set{H}_2$.}
  \label{fig:All-CSS-Tanner-graphs_8_4_d2}
\end{figure*}
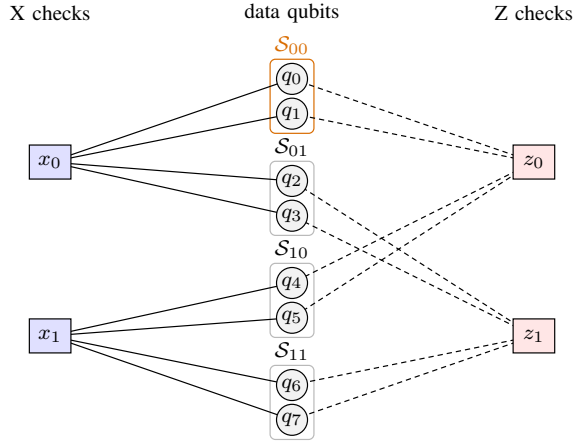
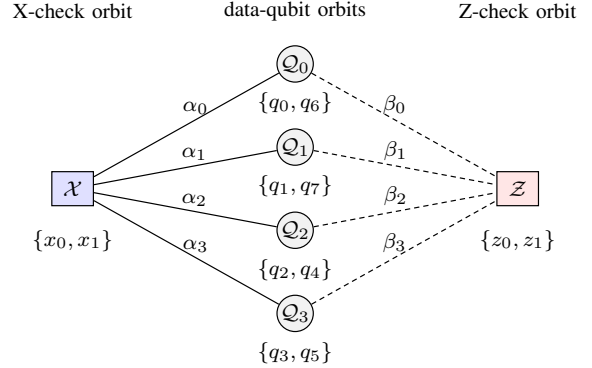

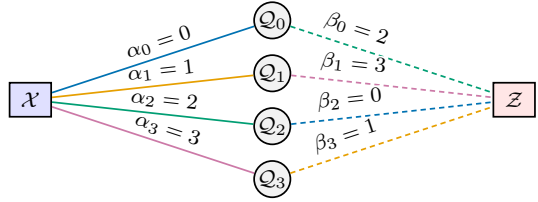
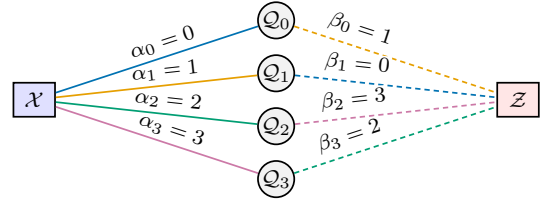
\begin{figure*}[t!]
  \centering
  \definecolor{orbitTick0}{RGB}{0,114,178}
  \definecolor{orbitTick1}{RGB}{230,159,0}
  \definecolor{orbitTick2}{RGB}{0,158,115}
  \definecolor{orbitTick3}{RGB}{204,121,167}
  \newcommand{\quotientColorGraph}[1]{%
    \begin{tikzpicture}[x=1cm,y=1cm,font=\fontsize{8}{9}\selectfont,line width=0.65pt]
      \node[rectangle,draw=black,fill=blue!12,minimum width=5.5mm,minimum height=4.5mm,inner sep=0pt] (CX) at (-3.2,0) {$\set{X}$};
      \node[rectangle,draw=black,fill=red!10,minimum width=5.5mm,minimum height=4.5mm,inner sep=0pt] (CZ) at (3.2,0) {$\set{Z}$};
      \foreach \k/\b in {#1} {
        \node[circle,draw=black,fill=gray!10,minimum size=4.8mm,inner sep=0pt] (CQ\k) at (0,{1.05-0.7*\k}) {$\set{Q}_{\k}$};
        \draw[orbitTick\k] (CX) -- node[pos=0.55,sloped,above=2pt,text=black,inner sep=0pt] {$\alpha_{\k}=\k$} (CQ\k);
        \draw[orbitTick\b,dash pattern=on 2pt off 1.5pt] (CQ\k) -- node[pos=0.30,sloped,allow upside down,above=2.5pt,text=black,fill=white,inner sep=0.5pt] {$\beta_{\k}=\b$} (CZ);
      }
    \end{tikzpicture}
  }
  \begin{minipage}[t]{0.48\textwidth}
    \centering
    \subfloat[$(\kappa_0,\kappa_1)=(0,0)$: ASC-feasible after lifting.]{\quotientColorGraph{0/2,1/3,2/0,3/1}
      \label{fig:quotient-colorings_8_4_d2-valid}}
    \par\nobreak\smallskip
    {\fontsize{8}{10}\selectfont $(\eI{\alpha_l<\beta_l})_{l=0}^3=(1,1\,|\,0,0)$.\par $\kappa_0=1+1=0$, $\kappa_1=0+0=0$ in $\Field_2$.\par}
  \end{minipage}\hfill
  \begin{minipage}[t]{0.48\textwidth}
    \centering
    \subfloat[$(\kappa_0,\kappa_1)=(1,1)$: quotient-only feasible.]{\quotientColorGraph{0/1,1/0,2/3,3/2}
      \label{fig:quotient-colorings_8_4_d2-invalid}}
    \par\nobreak\smallskip
    {\fontsize{8}{10}\selectfont $(\eI{\alpha_l<\beta_l})_{l=0}^3=(1,0\,|\,1,0)$.\par $\kappa_0=1+0=1$, $\kappa_1=1+0=1$ in $\Field_2$.\par}
  \end{minipage}
  \par\medskip
  \begin{tikzpicture}[font=\fontsize{8}{9}\selectfont,line width=0.8pt]
    \foreach \k in {0,1,2,3} {
      \draw[orbitTick\k] ({1.5*\k},0) -- ++(0.35,0);
      \node[anchor=west,inner sep=2pt] at ({1.5*\k+0.35},0) {tick $\k$};
    }
  \end{tikzpicture}
  \caption{Two depth-$4$ colorings of the quotient graph in Fig.~\ref{fig:All-CSS-Tanner-graphs_8_4_d2}~\protect\subref{fig:quotient-TG_8_4_d2}, with $(\alpha_0,\alpha_1,\alpha_2,\alpha_3)=(0,1,2,3)$. Edge colors denote ticks, not orbit classes; numerical labels specify the same assignments. Solid and dashed edges carry $\alpha_l$ and $\beta_l$, respectively. Both colorings are proper and satisfy $\kappa_0+\kappa_1=0$ in $\Field_2$. (a) With $(\beta_0,\beta_1,\beta_2,\beta_3)=(2,3,0,1)$, both original-graph parities vanish, so the lifted schedule is ASC-feasible. (b) With $(\beta_0,\beta_1,\beta_2,\beta_3)=(1,0,3,2)$, both parities are odd, so the quotient accepts an assignment that violates Definition~\ref{def:orbit-reduced-ASC}. The separator in each comparison vector groups $l=0,1$ for $\kappa_0$ and $l=2,3$ for $\kappa_1$.} 
  \label{fig:quotient-colorings_8_4_d2}
\end{figure*}

\section{Proof of Theorem~\ref{thm:orbit-invariant-schedule-lifting}}
\label{app:proof-orbit-lift}

Before proving Theorem~\ref{thm:orbit-invariant-schedule-lifting}, we present an example illustrating how constraints~\eqref{eq:properEC_orbit} and~\eqref{eq:q-ordering_orbit} in Definition~\ref{def:orbit-reduced-ASC} rule out quantum-invalid schedules.

\subsection{Proper Edge Coloring in \texorpdfstring{$\set{O}$}{O} is Not Enough}
\label{sec:NOT-enough_proper-edge-coloring}

Consider an $[[8,4,2]]$ code with $w=4$ and $\Delta_{\tn{D}}=2$, given by
\begin{IEEEeqnarray*}{rCl}
  \mat{H}_\tn{X}=
  \begin{pmatrix}
    1&1&1&1&0&0&0&0
    \\
    0&0&0&0&1&1&1&1
  \end{pmatrix},
  \\
  \mat{H}_\tn{Z}=
  \begin{pmatrix}
    1&1&0&0&1&1&0&0
    \\
    0&0&1&1&0&0&1&1
  \end{pmatrix}.
\end{IEEEeqnarray*}  
The Tanner graph of this code has $\ecard{\Aut(\graph{G})}=64$ automorphisms. Fig.~\ref{fig:All-CSS-Tanner-graphs_8_4_d2}~\protect\subref{fig:CSS-Tanner-graph_8_4_d2} shows its standard CSS Tanner graph. We next consider an automorphism illustrating how treating the quotient as a new code collapses the distinct global parity constraints preserved by exact orbit substitution into a single aggregate parity constraint on the quotient graph.

\begin{description}[style=unboxed,leftmargin=0pt,labelsep=0.5em]
\item[Step 1:] \emph{Choose a locally injective action.} Consider the automorphism
  \begin{IEEEeqnarray}{c}
    g_2=(x_0x_1)(z_0z_1)(q_0q_6)(q_1q_7)(q_2q_4)(q_3q_5),\,
    \set{H}_2=\egen{g_2}.
    \label{eq:8_4_d2-good-rotation}\IEEEeqnarraynumspace
  \end{IEEEeqnarray}
  Every edge is paired with its image under $g_2$, so $16=2\times 8$ and $\ecard{\set{O}}=8$. Both endpoints of every edge move under $g_2$. Thus, the two edges in each orbit share neither endpoint. Hence no orbit contains two edges incident on the same vertex.
  
\item[Step 2:] \emph{Name the eight orbit variables.} The following table lists every orbit. The symbols $\alpha_\ell$ and $\beta_\ell$ are simply short names for the corresponding $\tau_o$.

\vspace{2ex}
  \begin{center}
    \begin{tabular}{cll}
      \toprule
      $\ell$ & Orbit with tick $\alpha_{\ell}$ & Orbit with tick $\beta_{\ell}$\\
      \midrule
      $0$ & $\{(x_0,q_0),(x_1,q_6)\}$ & $\{(z_0,q_0),(z_1,q_6)\}$\\
      $1$ & $\{(x_0,q_1),(x_1,q_7)\}$ & $\{(z_0,q_1),(z_1,q_7)\}$\\
      $2$ & $\{(x_0,q_2),(x_1,q_4)\}$ & $\{(z_1,q_2),(z_0,q_4)\}$\\
      $3$ & $\{(x_0,q_3),(x_1,q_5)\}$ & $\{(z_1,q_3),(z_0,q_5)\}$\\
      \bottomrule
    \end{tabular}
  \end{center}
  \vspace{2ex}
  
  For example, $t_{(x_0,q_2)}=t_{(x_1,q_4)}=\alpha_2$. Applying \eqref{eq:properEC_orbit} at all original vertices gives exactly
  \begin{IEEEeqnarray}{c}
    \alpha_l\ne\alpha_\ell,\,\beta_l\ne\beta_\ell\, (0\le l<\ell\le3),\,\alpha_{\ell}\ne\beta_{\ell},\, \ell\in [0:3].
    \IEEEeqnarraynumspace\label{eq:properEC_8_4_d2}
  \end{IEEEeqnarray}
  The first two families come from the check vertices and the third from the data qubits.
  
\item[Step 3:] \emph{Expand each original quantum-ordering constraint.} All constraints derived from~\eqref{eq:q-ordering_orbit} can be written as follows:
  \begin{IEEEeqnarray}{rCl}
    \eta_{00}(\vect{t})& = &\eI{t_{(x_0,q_0)}<t_{(z_0,q_0)}}+\eI{t_{(x_0,q_1)}<t_{(z_0,q_1)}}
    \nonumber\\
    & = &\eI{\alpha_0<\beta_0}+\eI{\alpha_1<\beta_1}\eqdef\kappa_0,
    \nonumber\\[1mm]
    \eta_{01}(\vect{t})& = &\eI{t_{(x_0,q_2)}<t_{(z_1,q_2)}}+\eI{t_{(x_0,q_3)}<t_{(z_1,q_3)}}
    \nonumber\\[1mm]
    & = &\eI{\alpha_2<\beta_2}+\eI{\alpha_3<\beta_3}\eqdef\kappa_1,
    \nonumber\\[1mm]
    \eta_{10}(\vect{t})& = &\eI{t_{(x_1,q_4)}<t_{(z_0,q_4)}}+\eI{t_{(x_1,q_5)}<t_{(z_0,q_5)}}
    \nonumber\\[1mm]
    & = &\eI{\alpha_2<\beta_2}+\eI{\alpha_3<\beta_3}=\kappa_1,
    \nonumber\\[1mm]
    \eta_{11}(\vect{t})& = &\eI{t_{(x_1,q_6)}<t_{(z_1,q_6)}}+\eI{t_{(x_1,q_7)}<t_{(z_1,q_7)}}
    \nonumber\\
    & = &\eI{\alpha_0<\beta_0}+\eI{\alpha_1<\beta_1}=\kappa_0.
  \end{IEEEeqnarray}
  Thus, \eqref{eq:q-ordering_orbit} requires $\kappa_0=0$, $\kappa_1=0$, $\kappa_1=0$, $\kappa_0=0$. Removing identical whole equations leaves
  \begin{equation}
    \kappa_0=0\quad\tn{and}\quad \kappa_1=0.
    \label{eq:8_4_d2-two-parities}
  \end{equation}
  After orbit substitution, the check pairs $(x_0,z_0)$ and $(x_1,z_1)$ yield the same parity constraint, as do $(x_0,z_1)$ and $(x_1,z_0)$.
  
\item[Step 4:] \emph{Compare with the constraint generated on the quotient.}
  The vertex quotient has one $\tn{X}$ vertex, one $\tn{Z}$ vertex, and four data vertices $\set{Q}_0=\{q_0,q_6\}$, $\set{Q}_1=\{q_1,q_7\}$, $\set{Q}_2=\{q_2,q_4\}$, and $\set{Q}_3=\{q_3,q_5\}$. Each $\set{Q}_{\ell}$ is incident on an $\tn{X}$ edge with tick $\alpha_{\ell}$ and a $\tn{Z}$ edge with tick $\beta_{\ell}$, as shown in Fig.~\ref{fig:All-CSS-Tanner-graphs_8_4_d2}~\protect\subref{fig:quotient-TG_8_4_d2}. The check-orbit vertices are $\set{X}=\{x_0,x_1\}$ and $\set{Z}=\{z_0,z_1\}$. Applying Definition~\ref{def:asc-constraints} directly to the quotient yields only one parity constraint for its single check pair:
  \begin{equation}
    \sum_{l=0}^3\eI{\alpha_l<\beta_l}=\kappa_0+\kappa_1=0\pmod2.
    \label{eq:8_4_d2-quotient-parity}
  \end{equation}
  This equation allows both $(\kappa_0,\kappa_1)=(0,0)$ and $(\kappa_0,\kappa_1)=(1,1)$, whereas Definition~\ref{def:orbit-reduced-ASC} allows only $(0,0)$.

\item[Step 5:] \emph{Exhibit one valid assignment and one false positive.} At $\lambda=4$, fix $(\alpha_0,\alpha_1,\alpha_2,\alpha_3)=(0,1,2,3)$. Both choices below satisfy~\eqref{eq:properEC_8_4_d2}:
  
  \vspace{2ex}
  \begin{center}
    \footnotesize
    \setlength{\tabcolsep}{3pt}
    \renewcommand{\arraystretch}{1.1}
    \begin{tabularx}{\linewidth}{@{}>{\centering\arraybackslash}p{0.23\linewidth}>{\centering\arraybackslash}p{0.27\linewidth}>{\centering\arraybackslash}p{0.12\linewidth}>{\centering\arraybackslash}X@{}}
      \toprule
      $(\beta_0,\beta_1,\beta_2,\beta_3)$
      & $(\eI{\alpha_l<\beta_l})_{l=0}^3$ & $(\kappa_0,\kappa_1)$ & Original graph\\
      \midrule
      $(2,3,0,1)$ & $(1,1,0,0)$ & $(0,0)$ & ASC-feasible
      \\
      $(1,0,3,2)$ & $(1,0,1,0)$ & $(1,1)$ & \shortstack{violates all four parities}
      \\
      \bottomrule
    \end{tabularx}
  \end{center}
  \vspace{2ex}
  
  The first choice yields a depth-optimal schedule since $w=4$. The second passes~\eqref{eq:8_4_d2-quotient-parity}, but each of its two separate odd parities violates~\eqref{eq:8_4_d2-two-parities}. This discrepancy is specific to the quotient-code construction and does not arise under the exact orbit-reduced formulation of Definition~\ref{def:orbit-reduced-ASC}.
\end{description}

\subsection{The Proof}
\label{sec:the-proof}

We prove the two directions separately.

\begin{description}[style=unboxed,leftmargin=0pt,labelsep=0.5em]
\item[Forward direction.] Let $\tau$ satisfy Definition~\ref{def:orbit-reduced-ASC} and set $t_e=\tau_{\pi(e)}$ for every original edge $e$.

First, all coordinates of $\vect{t}$ lie in $[0:\lambda-1]$. For any original vertex $v$ and any distinct $e,e'\in\delta(v)$, \eqref{eq:properEC_orbit} gives
\begin{equation*}
  t_e=\tau_{\pi(e)}\ne\tau_{\pi(e')}=t_{e'}.
\end{equation*}
This is exactly~\eqref{eq:properEC_ASC}, simultaneously at the check vertices and the data vertices.

Second, for each original pair $(x_i,z_j)$ with $\set{S}_{ij}\ne\emptyset$, coordinate substitution gives
\begin{IEEEeqnarray*}{rCl}
  \eta_{ij}(\vect{t})& = &\sum_{q\in\set{S}_{ij}}\eI{t_{(x_i,q)}<t_{(z_j,q)}}\pmod 2
  \\
  & = &\sum_{q\in\set{S}_{ij}}\eI{\tau_{\pi(x_i,q)}<\tau_{\pi(z_j,q)}}\pmod 2 = 0,
\end{IEEEeqnarray*}
where the last equality is~\eqref{eq:q-ordering_orbit} for this same pair. Thus, \eqref{eq:q-ordering_ASC} holds. By Definition~\ref{def:asc-constraints}, $\vect{t}$ is ASC-feasible.

Third, $\pi(h(e))=\pi(e)$ for every $h\in\set{H}$, so
\begin{equation*}
  t_{h(e)}=\tau_{\pi(h(e))}=\tau_{\pi(e)}=t_e.
\end{equation*}
Hence, the lifted schedule is $\set{H}$-invariant.

\item[Reverse direction.] Let $\vect{t}$ be an $\set{H}$-invariant schedule satisfying Definition~\ref{def:asc-constraints}. For each orbit $o$, choose any representative $e_o\in o$ and set $\tau_o\eqdef t_{e_o}$. Since $t_{e_o}\in [0:\lambda-1]$, every $\tau_o$ lies in the tick domain required by Definition~\ref{def:orbit-reduced-ASC}. If $e'_o\in o$ is another representative, then $e'_o=h(e_o)$ for some $h\in\set{H}$, and invariance gives $t_{e'_o}=t_{e_o}$. Thus $\tau_o$ is well defined and $t_e=\tau_{\pi(e)}$ for every edge.
  
  For every $v$ and distinct $e,e'\in\delta(v)$, \eqref{eq:properEC_ASC} now gives $\tau_{\pi(e)}=t_e\ne t_{e'}=\tau_{\pi(e')}$, which is \eqref{eq:properEC_orbit}. For each original overlapping check pair, the same term-by-term substitution into \eqref{eq:q-ordering_ASC} gives \eqref{eq:q-ordering_orbit}. Therefore $\tau$ satisfies Definition~\ref{def:orbit-reduced-ASC}. Its value on each orbit is forced by $t_{e_o}$, proving uniqueness. The two constructions are inverse.
\end{description}



\section{Proof of Proposition~\ref{prop:automorphisms-classical-Tan}}
\label{sec:proof_automorphisms-classical-Tan}

When $f$ is restricted to $\delta(v)$, the edges adjacent to a check node $v$ in the base graph $\graph{B}$, we can view it as a map
\begin{equation}
  \begin{tikzcd}
    [column sep=12.26mm]
    \set{V}_\tn{d}^\mat{H} \ar[r, "(\mu|_{\delta(v)})^{-1}", "\cong"{swap}] & \delta(v) \ar[r, "f|_{\delta(v)}"] & \delta(f(v)) \ar[r, "\mu|_{\delta(f(v))}", "\cong"{swap}] & \set{V}_\tn{d}^\mat{H}
  \end{tikzcd}
  \label{eq:local-var-map}
\end{equation}
on the variable nodes of $\graph{G}^\mat{H}$.
Recall that $((f(v), \sigma_v(t)), f(q))$ is an edge in $\graph{G}$ if and only if $(f(v), f(q)) \in \set{E}$ and $(\mu(f(v), f(q)), \sigma_v(t))\in \set{E}^\mat{H}$.
If we let $p = \mu|_{\delta(v)}(v, q)\in \set{V}_{\tn{d}}^\mat{H}$, then (\ref{eq:local-var-map}) maps $p$ to $\mu(f(v), f(q))$. Furthermore, if $((v, t), q)$ is an edge in $\graph{G}$, then $(p, t) = (\mu(v, q), t)\in \set{E}^\mat{H}$, which is mapped by $\sigma_v \sqcup f|_{\delta(v)}$ to $(\mu(f(v), f(q)), \sigma_v(t))$.
Therefore, an edge $((v, t), q)$ in $\graph{G}$ is mapped to an edge by $\tilde{f}_\sigma$ if and only if
$\sigma_v \sqcup f|_{\delta(v)}$ is an automorphism on $\graph{G}^{\mat{H}}$ for all $v\in \set{V}_{\tn{c}}$. 




\section{Proof of Corollary~\ref{cor:nice-automorphisms}}
\label{sec:proof_nice-automorphisms}

This is the special case where $\sigma$ is the identity. The graph automorphism $f$ is label-preserving if and only if the map (\ref{eq:local-var-map}) is the identity for all $v\in \set{V}_{\tn{c}}$, and the identity is an automorphism of $\graph{G}^{\mat{H}}$.

\section{Proof of Theorem~\ref{thm:achievable-depth_qTanner-local-codes}}
\label{sec:proof_achievable-depth_qTanner-local-codes}

Let $\vect{t}_{\tn{A}}:\set{E}_{\graph{G}(\mat{G}_{\tn{A}}, \mat{H}_{\tn{A}})} \to [0:\lambda_1-1]$ and $\vect{t}_\tn{B}:\set{E}_{\graph{G}(\mat{G}_\tn{B}, \mat{H}_\tn{B})} \to [0:\lambda_2-1]$ be the given quantum-valid ordered edge colorings. 
Since the edges of ${\graph{G}(\mat{G}_\tn{A}\otimes \mat{G}_\tn{B}, \mat{H}_\tn{A} \otimes \mat{H}_\tn{B})}$ are precisely products of edges of ${\graph{G}(\mat{G}_\tn{A}, \mat{H}_\tn{A})}$ and ${\graph{G}(\mat{G}_\tn{B}, \mat{H}_\tn{B})}$, the product of the two colorings is of the form
\begin{equation*}
  \vect{t}_\tn{A}\times \vect{t}_\tn{B} : \set{E}_{\graph{G}(\mat{G}_\tn{A}\otimes \mat{G}_\tn{B}, \mat{H}_\tn{A} \otimes \mat{H}_\tn{B})} \to [0:\lambda_1-1]\times [0:\lambda_2-1].
\end{equation*}
If we order $[0:\lambda_1-1]\times [0:\lambda_2-1]$ row-wise, meaning $(i,j)<(i',j')$ when $i<i'$ or $i=i'$ and $j<j'$, then $\vect{t}_\tn{A}\times \vect{t}_\tn{B}$ becomes an ASC-feasible schedule for  $\graph{G}(\mat{G}_\tn{A} \otimes \mat{G}_\tn{B}, \mat{H}_\tn{A} \otimes \mat{H}_\tn{B})$.

The quantum-ordering constraint~\eqref{eq:q-ordering_ASC} is satisfied for $\vect{t}_\tn{A}\times \vect{t}_\tn{B}$ also when restricting the variable nodes to any column/row of $\set{A}\times \set{B}$, since $\vect{t}_\tn{A}$ and $\vect{t}_\tn{B}$ both satisfy it. Therefore,~\eqref{eq:q-ordering_ASC} is also satisfied when lifting $\vect{t}_\tn{A}\times \vect{t}_\tn{B}$ to the quantum Tanner code: when the local view of an $\tn{X}$-check node and a $\tn{Z}$-check node have nonzero intersection in the Tanner graph, they come from neighboring vertices of $\set{X}$, meaning that the intersection is one or more rows or columns in both the local views.

Lastly, we need to check~\eqref{eq:properEC_ASC} to ensure that we have a proper edge coloring. By construction, each check node is connected only to edges of different colors. A variable node is connected to four check vertices, two $\tn{X}$-checks and two $\tn{Z}$-checks. If one of the $\tn{X}$-checks' color is given by the $(a, b)$-edge of $\mat{G}_\tn{A}\otimes \mat{G}_\tn{B}$, then the other $\tn{X}$-check is the $(a^{-1}, b^{-1})$-edge of $\mat{G}_\tn{A}\otimes \mat{G}_\tn{B}$ and the two $\tn{Z}$-checks are the $(a, b^{-1})$-edge and $(a^{-1}, b)$-edge of $\mat{H}_\tn{A} \otimes \mat{H}_\tn{B}$. The $\tn{X}$- and $\tn{Z}$-edges have different colors since $\vect{t}_\tn{A}$ and $\vect{t}_\tn{B}$ are both ASC-feasible. The two $\tn{X}$-edges share no colors as $b$-edges and $b^{-1}$-edges share no colors, and the same is true for the two $\tn{Z}$-edges.

\section{Performance of the $[[72,12,6]]$ BB Code Under Circuit-Level Noise}
\label{sec:LocalASC_versus_IBM_BB_72_12_d6}

In Fig.~\ref{fig:72_12_schedule-comparison}, we compare the memory experiment performance, as described in Section~\ref{sec:numerical-results}, of a syndrome-extraction schedule obtained from LocalASC to that of the official schedule for BB codes~\cite{Bravyi-etal24_1}. This is done by recreating the BB code circuits in \texttt{stim} and running memory experiments on both circuits, using BP+OSD-CS. However, we use fewer BP iterations ($1000$ vs $10000$) and a higher CS order ($10$ vs $7$ in their released code), leading to differences in decoding performance for the same code compared with~\cite[Fig.~3]{Bravyi-etal24_1}.

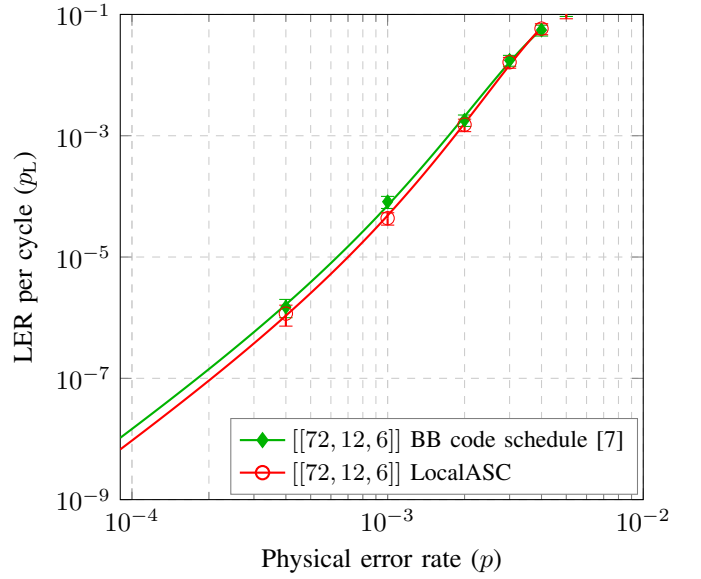
\begin{figure}[t!]
  \centering
  \begin{tikzpicture}
    \begin{loglogaxis}[
      width=8.5cm,
      height=8cm,
      xmin=9e-5, xmax=0.01,
      ymin=1e-9, ymax=1e-1,
      xlabel={Physical error rate ($p$)},
      ylabel={LER per cycle ($p_\tn{L}$)},
      grid=both,
      grid style={dashed, draw=gray!40},      
      xtick={1e-4, 1e-3, 1e-2},
      ytick={1e-9, 1e-7, 1e-5, 1e-3, 1e-1},
      legend pos=south east,
      legend style={
        at={(0.99,0.01)},
        anchor=south east,
        font=\small,
        cells={anchor=west},
        draw=gray},      
      ]
      
      
      \addplot [
      color=green!70!black, 
      thick, mark=none,
      legend image post style={mark=diamond*, mark size=2.5pt}
      ] table [x=p, y=pL, col sep=comma] {plot_data/72_12_6_Bravyi_fit.csv};
      \addlegendentry{$[[72, 12, 6]]$ BB code schedule~\cite{Bravyi-etal24_1}}
      
      \addplot [
      color=green!70!black,
      only marks, forget plot,
      mark=diamond*, mark size=2.5pt,
      error bars/.cd,
      y dir=both, y explicit,
      error bar style={thick, solid}
      ] table [
      x=p, 
      y=total_pL, 
      y error=total_pL_err, 
      col sep=comma
      ]{plot_data/72_12_6_Bravyi_data.csv};
      
      \addplot [
      color=red, 
      thick, mark=none,
      legend image post style={mark=o, mark size=2.5pt}
      ] table [x=p, y=pL, col sep=comma] {plot_data/72_12_6_LocalASC_fit.csv};
      \addlegendentry{$[[72, 12, 6]]$ LocalASC}
      
      \addplot [
      color=red,
      only marks, forget plot,
      mark=o, mark size=2.5pt,
      error bars/.cd,
      y dir=both, y explicit,
      error bar style={thick, solid}
      ] table [
      x=p, 
      y=total_pL, 
      y error=total_pL_err, 
      col sep=comma
      ]{plot_data/72_12_6_LocalASC_data.csv};      
    \end{loglogaxis}
  \end{tikzpicture}
  \vspace{-3ex}
  \caption{We produce $\tn{Z}$- and $\tn{X}$-memory circuits with each schedule. These are decoded using BP+OSD-CS10 with a maximum of $1000$ BP iterations and at most $10^6$ samples per noise level, stopping each experiment once it reaches $10$ logical errors. For the extrapolation, we consider $d_{\tn{circ}} = 6$ from~\cite[Table 1]{Bravyi-etal24_1}.}
  \label{fig:72_12_schedule-comparison}
\end{figure}


\IEEEtriggeratref{20}
\bibliographystyle{IEEEtran}
\bibliography{defshort1, biblioHY}

\end{document}